\documentclass[letterpaper,11pt]{article}
\usepackage[margin=1in]{geometry}
\usepackage{setspace}
\usepackage{graphicx}
\usepackage{epsfig}

\usepackage{microtype} %
\usepackage{array} %

\usepackage{amssymb}
\usepackage{amsmath}
\usepackage{amsthm}
\usepackage{mathtools} %
\usepackage{stmaryrd} %
\usepackage{stackrel} %
\usepackage{latexsym} %

\usepackage{tikz}
\usetikzlibrary{arrows}
\usetikzlibrary{arrows.meta}
\usetikzlibrary{decorations.pathreplacing}
\usetikzlibrary{cd}
\usepackage[all]{xy} %

\usepackage[section]{algorithm}
\usepackage[noend]{algpseudocode}
\algnewcommand{\LineComment}[1]{\Statex\hspace{\algorithmicindent}\(\triangleright\) #1}
\algnewcommand\algorithmicforeach{\textbf{for each}}
\algdef{S}[FOR]{ForEach}[1]{\algorithmicforeach\ #1\ \algorithmicdo}
\usepackage{etoolbox} %

\usepackage{caption}
\usepackage{subcaption}

\usepackage{makecell} %
\usepackage{booktabs} %
\usepackage{multirow} %

\usepackage[hidelinks]{hyperref} %
\usepackage[symbol]{footmisc}
\usepackage{makeidx} %
\usepackage{url} %
\usepackage{color}
\usepackage{xcolor}
\usepackage{xspace} %
\usepackage[normalem]{ulem} %
\usepackage{soul} %
\usepackage{ifpdf} %
\usepackage{extarrows}

\makeatletter
\expandafter\patchcmd\csname\string\algorithmic\endcsname{\itemsep\z@}{\itemsep=0.25ex}{}{}
\makeatother

\newcounter{usesmallsep}
\ifnum\the\value{usesmallsep}=1

    \newlength{\myitemsep}
    \newlength{\mytopsep}
    \usepackage{enumitem}
    \setlist[itemize]{leftmargin=\parindent,parsep=\parskip,
      listparindent=\parindent,itemsep=\myitemsep,topsep=\myitemsep}
    \setlist[enumerate]{leftmargin=\parindent,parsep=\parskip,
      listparindent=\parindent,itemsep=\myitemsep,,topsep=\myitemsep}
    \setlist[description]{font=\bfseries,leftmargin=\parindent,parsep=\parskip,
      listparindent=\parindent,itemsep=\myitemsep,topsep=\myitemsep}
    
    \newlength{\mypartitlesep}
    \usepackage[explicit]{titlesec}
    \titlespacing{\paragraph}{0pt}{\mypartitlesep}{\mypartitlesep}
    
    \newlength{\mythmsep}
    \newtheoremstyle{mythmstyle}
      {\mythmsep} %
      {\mythmsep} %
      {\itshape} %
      {} %
      {\bfseries} %
      {.} %
      {.5em} %
      {} %
    
    \newtheoremstyle{mydefstyle}
      {\mythmsep} %
      {\mythmsep} %
      {} %
      {} %
      {\bfseries} %
      {.} %
      {.5em} %
      {} %
    
    \theoremstyle{mythmstyle}
        \newtheorem{theorem}{Theorem}
        \newtheorem{proposition}[theorem]{Proposition}
        \newtheorem{lemma}[theorem]{Lemma}
        \newtheorem{corollary}[theorem]{Corollary}
        \newtheorem{fact}[theorem]{Fact}
        \newtheorem*{fact*}{Fact}
    \theoremstyle{mydefstyle}
        \newtheorem{definition}{Definition}
        \newtheorem{problem}{Problem}
        \newtheorem{assumption}{Assumption}
        \newtheorem{remark}{Remark}
        \newtheorem{algr}[algorithm]{Algorithm}

    \newenvironment{proof}
        {\vspace{-0.9em}\begin{proof}}
        {\end{proof}\vspace{-0.4em}}

\else

    \theoremstyle{plain}
        \newtheorem{theorem}{Theorem}
        \newtheorem{proposition}[theorem]{Proposition}

        \newtheorem{fact}[theorem]{Fact}

        \newtheorem*{algr*}{Algorithm}
    \theoremstyle{definition}
        \newtheorem{definition}[theorem]{Definition}

        \newtheorem{remark}[theorem]{Remark}
        
    \usepackage{enumitem}
    \setlist[itemize]{leftmargin=\parindent}
    \setlist[enumerate]{leftmargin=\parindent}
    \setlist[description]{font=\bfseries,leftmargin=\parindent}

\fi

\newcommand{\Hm}{\mathsf{H}}

\newcommand{\Real}{\mathbb{R}}

\newcommand{\fsimp}[2]{\sigma_{#2}}
\newcommand{\fnrsimp}[2]{\hat{\sigma}_{#2}}

\newcommand{\id}{\mathtt{id}}

\newcommand{\Pers}{\mathsf{Pers}}

\newcommand{\lbarrowspace}{\;}

\let\leftrightarrowsp\lrarrowsp

\newcommand{\incto}{\hookrightarrow}
\newcommand{\inctosp}[1]{\xhookrightarrow{\lbarrowspace#1\lbarrowspace}}
\newcommand{\bakincto}{\hookleftarrow}
\newcommand{\bakinctosp}[1]{\xhookleftarrow{\lbarrowspace#1\lbarrowspace}}

\newcommand{\given}{\,|\,}
\newcommand{\Set}[1]{\{#1\}}

\let\emptyset\varnothing

\let\intsec\intersect
\let\union\cup

\newcommand{\Ecal}{\mathcal{E}}

\newcommand{\Fcal}{\mathcal{F}}

\newcommand{\Ical}{\mathcal{I}}

\newcommand{\Lcal}{\mathcal{L}}

\newcommand{\Ucal}{\mathcal{U}}

\newcommand{\Zbb}{\mathbb{Z}}

\newcommand{\DG}{\Delta}

\newcommand{\LG}{\Lambda}

\newcommand{\oG}{\omega}

\newcommand{\sG}{\sigma}

\newcommand{\tG}{\tau}

\newcommand{\Dim}{p}

\newcommand{\birth}{b}
\newcommand{\death}{d}
\newcommand{\filtcnt}{m}
\newcommand{\simpcnt}{n}

\newcommand{\splx}{\sigma}
\newcommand{\Fnr}{{\hat{\Fcal}}}
\newcommand{\ud}{\Ucal}
\newcommand{\ucplx}{L}
\newcommand{\usimp}{\tau}
\newcommand{\ef}{\Ecal}
\newcommand{\add}[1]{{\downarrow}#1}
\newcommand{\del}[1]{{\uparrow}#1}

\newcommand{\Phat}{\textsc{Phat}}
\newcommand{\Gudhi}{\textsc{Gudhi}}
\newcommand{\Dionysustwo}{\textsc{Dionysus2}}

\newcommand{\cancel}[1]

\begin{document}

\title{Fast Computation of Zigzag Persistence\thanks{This research is partially supported by NSF grant CCF 2049010.}}

\author{Tamal K. Dey\thanks{Department of Computer Science, Purdue University. \texttt{tamaldey@purdue.edu}}
\and Tao Hou\thanks{School of Computing, DePaul University. \texttt{taohou01@gmail.com}}
}

\date{}

\maketitle
\thispagestyle{empty}

\begin{abstract}
Zigzag persistence is a powerful extension of the standard persistence
which allows deletions of simplices besides
insertions.
However, 
computing
zigzag persistence usually takes considerably more time than the standard persistence.
We propose an algorithm called \textsc{FastZigzag} which
narrows this efficiency gap.
Our main result is that an input simplex-wise zigzag filtration can be
converted to a \emph{cell}-wise non-zigzag filtration of a {$\DG$-complex} 
with the same length,
where the cells are copies of the input simplices. This conversion
step in \textsc{FastZigzag} incurs very little cost. Furthermore, the barcode
of the original filtration can be easily read from the barcode of the
new cell-wise filtration
because the conversion embodies
a series of \emph{diamond switches} known in topological data analysis.
This seemingly simple observation
opens up the vast possibilities for improving the computation of zigzag persistence
because any efficient algorithm/software for standard persistence
can now be applied to computing
zigzag persistence.
Our experiment shows that
this indeed achieves substantial performance gain over the existing state-of-the-art softwares.
\end{abstract}

\newpage
\setcounter{page}{1}

\newcommand{\defemph}[1]{\emph{#1}}

\section{Introduction}

Standard persistent homology defined over a growing sequence 
of simplicial complexes is
a fundamental tool in topological data analysis (TDA). Since the advent of
persistence algorithm~\cite{edelsbrunner2000topological} and its algebraic
understanding~\cite{zomorodian2005computing}, various extensions of the basic concept
have been explored~\cite{carlsson2010zigzag,carlsson2009zigzag-realvalue,cohen2009extending,de2011dualities}.
Among these extensions,
zigzag persistence introduced by Carlsson and de Silva~\cite{carlsson2010zigzag} is
an important one. It
empowered TDA to deal with 
filtrations where both insertion
and deletion of simplices are allowed. 
In practice, allowing deletion of simplices does make the topological tool more powerful.
For example, in dynamic networks~\cite{dey2021graph,holme2012temporal}
 a sequence of graphs may not grow monotonically but can also shrink due to
disappearance of vertex connections. Furthermore,
zigzag persistence seems to be naturally connected with the 
computations involving multiparameter
persistence, see e.g.~\cite{dey2021updating,dey2021computing}.

Zigzag persistence possesses some key differences
from standard persistence. For example,
unlike standard (non-zigzag) modules which decompose
into only finite and infinite intervals, zigzag modules
decompose into {four types} of intervals (see Definition~\ref{dfn:open-close-bd}).
Existing algorithms for computing zigzag persistence from a zigzag filtration~\cite{carlsson2009zigzag-realvalue,maria2014zigzag,maria2016computing,maria2019discrete}
are all based on maintaining explicitly or implicitly a consistent basis throughout the filtration.
This makes these algorithms
for zigzag persistence 
more involved and hence slower in practice
than algorithms for
the non-zigzag version though they have the same time complexity~\cite{milosavljevic2011zigzag}.
We sidestep the bottleneck of maintaining an explicit basis 
and propose an algorithm called \textsc{FastZigzag}, which converts
the input zigzag filtration
to a \emph{non-zigzag} filtration with an efficient strategy for mapping barcodes of the two bijectively. Then, we can apply any  efficient algorithm for standard persistence  on the
resulting non-zigzag filtration to compute the barcode of the original filtration.
Considering the abundance of optimizations~\cite{bauer2021ripser,bauer2014clear,bauer2017phat,BDM15,chen2011persistent,CK13} of standard persistence algorithms
and a recent GPU acceleration~\cite{zhang2020gpu},
the conversion in \textsc{FastZigzag} enables zigzag persistence computation
to take advantage of any existing or future improvements on standard persistence computation.
Our implementation, which uses the
\Phat{}~\cite{bauer2017phat} software for computing standard persistence,
shows substantial performance gain
over existing state-of-the-art softwares~\cite{Dionysus2,gudhi:urm}
for computing zigzag persistence
(see Section~\ref{sec:exp}).
We make our software publicly available through: 
\url{https://github.com/taohou01/fzz}.

To elaborate on the strategy of \textsc{FastZigzag}, we first observe 
a special type of zigzag filtrations 
called \emph{non-repetitive} zigzag filtrations in which
a simplex (or more generally, a \emph{cell}) is never added again once deleted. Such a filtration
admits an {\it up-down} filtration as its canonical form that can be
obtained by a series of \emph{diamond switches}~\cite{carlsson2010zigzag,carlsson2019parametrized,carlsson2009zigzag-realvalue}.
The up-down filtration can be further converted
into a {\it non-zigzag} filtration again using diamond switches as in the Mayer-Vietoris
pyramid presented in~\cite{carlsson2009zigzag-realvalue}. 
Individual switches are atomic tools that help us to show equivalence of barcodes, but
we do not need to actually execute them in computation.
Instead, we go straight
to the final form of the filtration quite easily and efficiently.
Finally, we observe that any zigzag filtration 
can be treated as a non-repetitive \emph{cell-wise} 
filtration of a \emph{$\DG$-complex}~\cite{hatcher2002algebraic} 
consisting of multisets of input simplices. 
This means that each repeatedly added
simplex is treated as a different cell in the $\DG$-complex, 
so that we can apply 
our findings for non-repetitive filtrations to
arbitrary filtrations.
The conversions
described above are detailed in Section~\ref{sec:fzz}.

\subsection{Related works}

Zigzag persistence is essentially an \emph{$A_n$-type quiver}~\cite{derksen2005quiver} in mathematics which
is first introduced to the TDA community by Carlsson and de Silva~\cite{carlsson2010zigzag}.
In their paper~\cite{carlsson2010zigzag},
Carlsson and de Silva also study the Mayer-Vietoris diamond used in this paper
and propose an algorithm for computing zigzag barcodes from zigzag modules
(i.e., an input is a sequence of vector spaces connected by linear maps encoded as
matrices).
Carlsson et al.~\cite{carlsson2009zigzag-realvalue} then propose an $O(mn^2)$ algorithm 
for computing zigzag barcodes from zigzag filtrations using a structure called right filtration.
In their paper~\cite{carlsson2009zigzag-realvalue}, Carlsson et al.\ also extend the classical sublevelset
filtrations for functions on topological spaces by proposing levelset zigzag filtrations
and show the equivalence of levelset zigzag with the extended persistence proposed by Cohen-Steiner et al.~\cite{cohen2009extending}.
Maria and Outdot~\cite{maria2014zigzag,maria2016computing}
propose an alternative algorithm for computing zigzag barcodes by 
attaching a reversed standard filtration to the end
of the partial zigzag filtration being scanned.
Their algorithm maintains the barcode over the Surjective and Transposition Diamond
on the constructed zigzag filtration~\cite{maria2014zigzag,maria2016computing}.
Maria and Schreiber~\cite{maria2019discrete} propose a Morse reduction preprocessing
for zigzag filtrations
which speeds up the zigzag barcode computation.
Carlsson et al.~\cite{carlsson2019persistent} discuss some matrix factorization techniques
for computing zigzag barcodes from zigzag modules, which, combined with a divide-and-conquer strategy, 
lead to a parallel algorithm for computing zigzag persistence.
Almost all algorithms reviewed so far have a cubic time complexity.
Milosavljevi{\'c} et al.~\cite{milosavljevic2011zigzag} establish
an $O(m^\oG)$ theoretical complexity for computing zigzag persistence from filtrations,
where $\omega< 2.37286$ is the matrix multiplication exponent~\cite{alman2021refined}.
Recently, Dey and Hou~\cite{dey2021graph} propose near-linear algorithms for computing 
zigzag persistence from the special cases of graph filtrations,
with the help of representatives defined for the intervals
and some dynamic graph data structures.

\section{Preliminaries}

\paragraph{$\DG$-complex.}

In this paper, we build filtrations on \emph{$\DG$-complexes}
which are extensions of simplicial complexes
described in Hatcher~\cite{hatcher2002algebraic}.
These $\DG$-complexes are
derived from a set of standard simplices by identifying the boundary
of each simplex with other simplices while preserving the vertex orders.
For distinction,
building blocks of $\DG$-complexes
(i.e., standard simplices) 
are called \emph{cells}.
Motivated by a construction 
from the input simplicial complex described in Algorithm~\ref{alg:conv-filt},
we use a \emph{more restricted} version of $\DG$-complexes,
where boundary cells of each $\Dim$-cell are identified with \emph{distinct} $(\Dim-1)$-cells.
Notice that this makes 
each $p$-cell 
combinatorially equivalent to
a $p$-simplex.
Hence, the difference of the $\DG$-complexes
in this paper
from the standard simplicial complexes
is that
common faces
of two cells in the $\DG$-complexes
can have more
relaxed forms.
For example,
in Figure~\ref{fig:bubble}, two `triangles' (2-cells)
in a $\DG$-complex having the same set of vertices
can either 
share 0, 1, 2, or 3 edges in their boundaries;
note that the two triangles in Figure~\ref{fig:two-tri-3}
form a 2-cycle.

\begin{figure}[t]
  \centering
  \captionsetup[subfigure]{justification=centering}

  \begin{subfigure}[t]{0.13\textwidth}
  \centering
  \includegraphics[width=\linewidth]{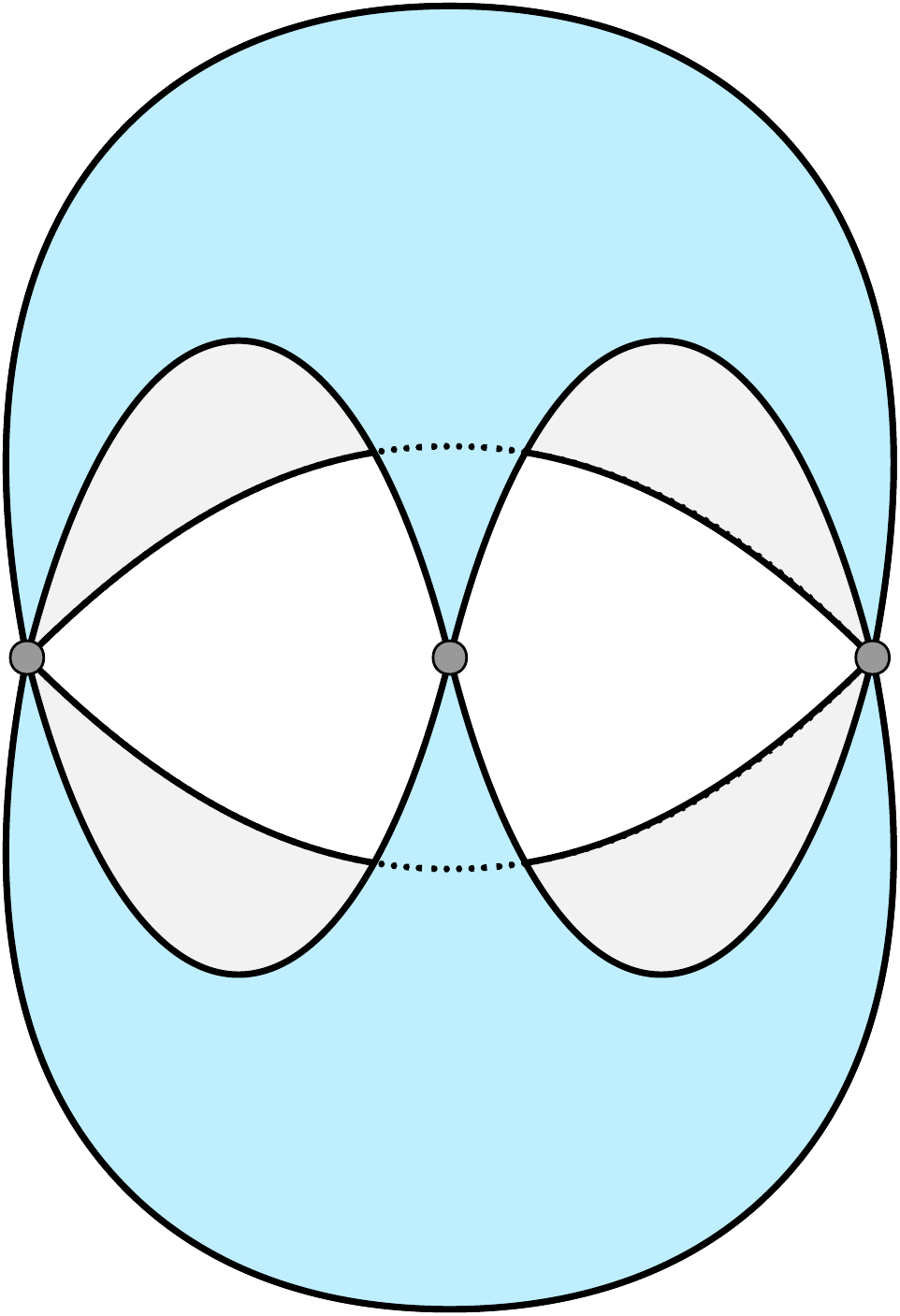}
  \caption{}
  \label{fig:two-tri-0}
  \end{subfigure}
  \hspace{4em}
  \begin{subfigure}[t]{0.13\textwidth}
  \centering
  \includegraphics[width=\linewidth]{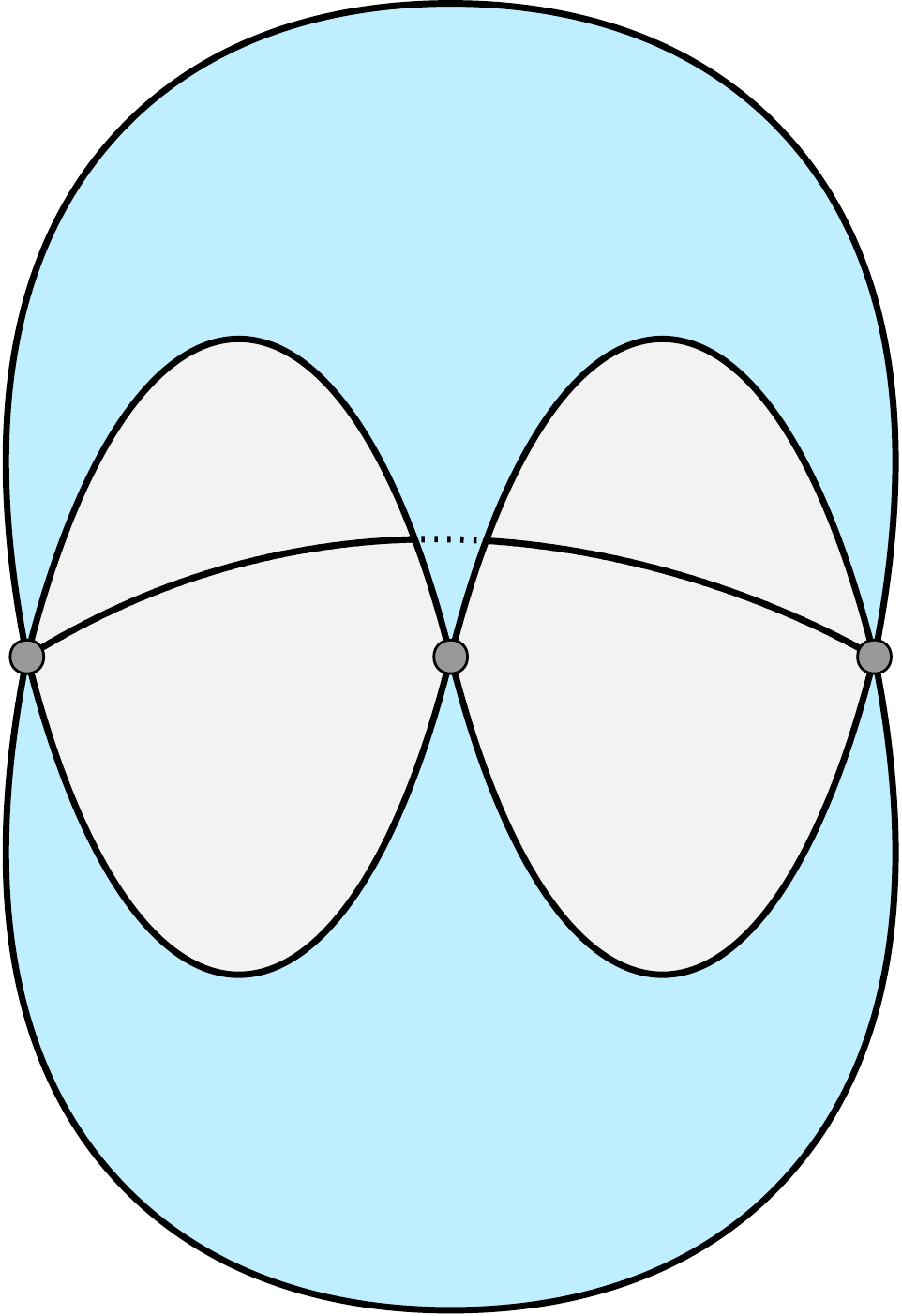}
  \caption{}
  \label{fig:two-tri-1}
  \end{subfigure}
  \hspace{4em}
  \begin{subfigure}[t]{0.13\textwidth}
  \centering
  \includegraphics[width=\linewidth]{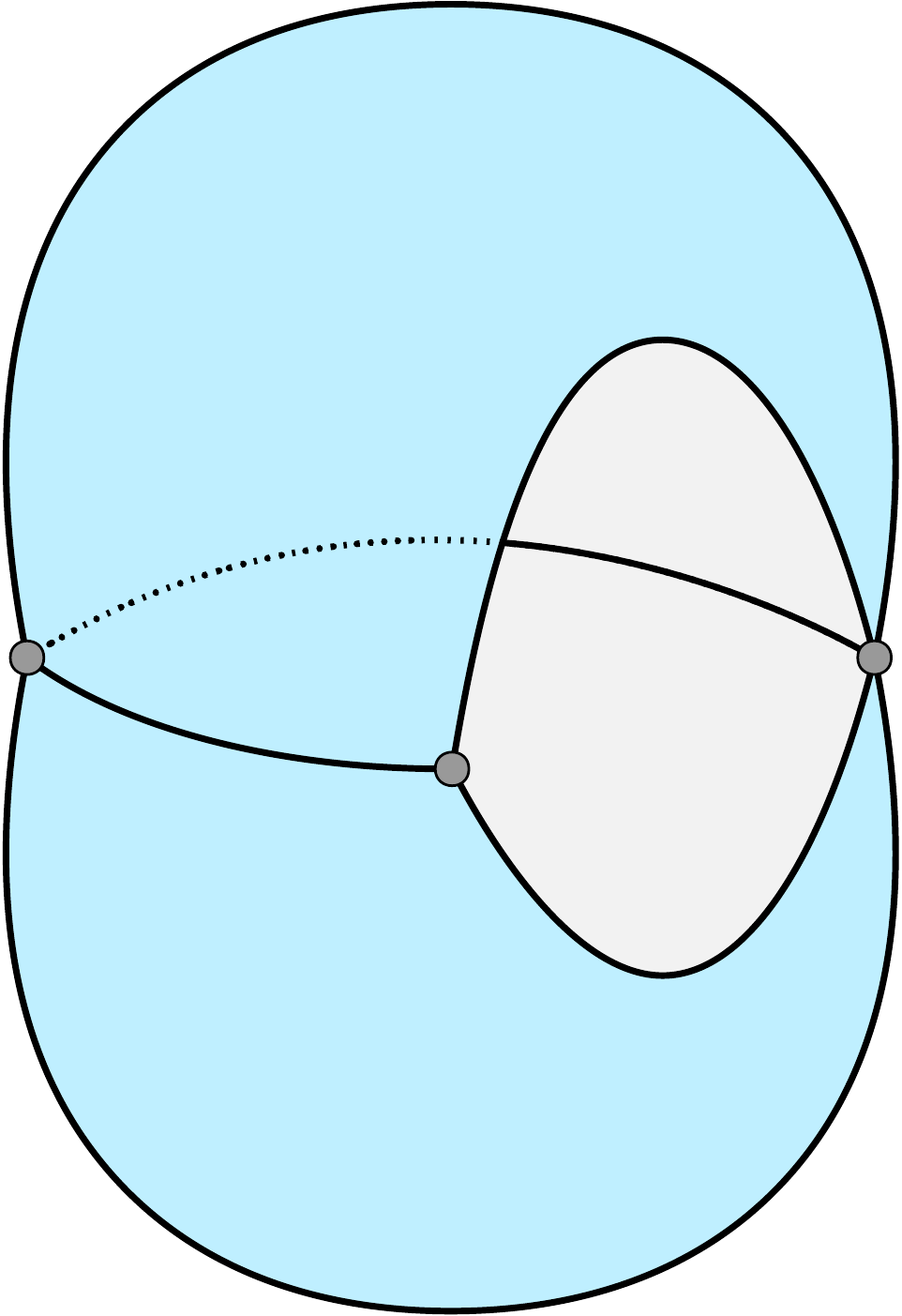}
  \caption{}
  \label{fig:two-tri-2}
  \end{subfigure}
  \hspace{4em}
  \begin{subfigure}[t]{0.13\textwidth}
  \centering
  \includegraphics[width=\linewidth]{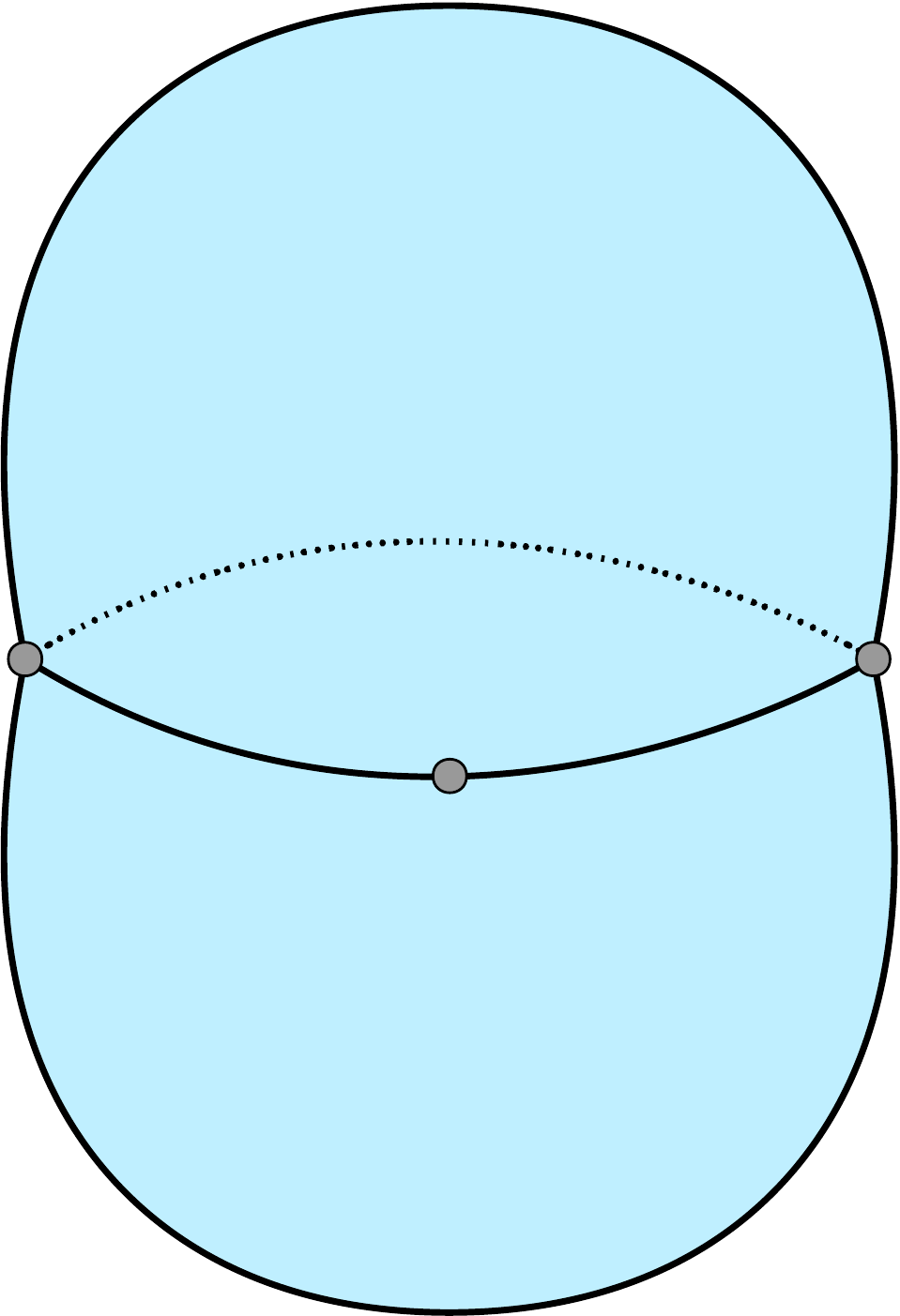}
  \caption{}
  \label{fig:two-tri-3}
  \end{subfigure}

  \caption{Examples of $\DG$-complexes with two triangles sharing 0, 1, 2, or 3 edges on their boundaries.}
  \label{fig:bubble}
\end{figure}

Formally, we define $\DG$-complexes recursively
similar to the classical definition of \emph{CW-complexes}~\cite{hatcher2002algebraic} 
though it need not be as general;
see Hatcher's book~\cite{hatcher2002algebraic} for a more general definition. 
Note that simplicial complexes are trivially 
$\DG$-complexes and therefore most definitions
in this section target $\DG$-complexes.

\begin{definition}
 A \defemph{$\DG$-complex} is defined recursively with dimension:
\begin{enumerate}
     \item A $0$-dimensional $\DG$-complex $K^0$ is a set of points, each called a \defemph{$0$-cell}.

     \item A $p$-dimensional $\DG$-complex $K^p$, $p\geq 1$, 
     is a quotient space of 
     a $(p-1)$-dimensional $\DG$-complex $K^{p-1}$ 
     along with several standard $p$-simplices. The quotienting 
     is realized by an attaching map
     $h:\partial(\sG)\to K^{p-1}$ which identifies the boundary $\partial(\sG)$ of each
     $p$-simplex $\sG$ with points in $K^{p-1}$ 
     so that $h$ is a homeomorphism onto its image.
     We term the standard $p$-simplex $\sG$ with boundary
     identified to $K^{p-1}$ as a \defemph{$p$-cell} in $K^p$.
     Furthermore, 
     we have that the restriction
     of $h$ to each proper face of $\sG$ is a \emph{homeomorphism} onto a cell in $K^{p-1}$.
 \end{enumerate}
\end{definition}

Notice that the original (more general) $\DG$-complexes~\cite{hatcher2002algebraic} require specifying vertex orders
when identifying the cells.
However, the restricted $\DG$-complexes defined above
do not require specifying such orders because
we always identify the boundaries
of cells by \emph{homeomorphisms} and hence the vertex orders for identification
are implicitly derived from a vertex order of a seeding cell.

\paragraph{Homology.}
Homology in this paper
is defined on $\DG$-complexes,
which is defined similarly as for simplicial complexes~\cite{hatcher2002algebraic}.
All homology groups are taken with $\Zbb_2$-coefficients
and therefore vector spaces mentioned in this paper are also over $\Zbb_2$.

\paragraph{Zigzag filtration and barcode.}

A {\it zigzag filtration} (or simply {\it filtration})
is a sequence of $\DG$-complexes 
\[\Fcal: K_0 \leftrightarrow K_1 \leftrightarrow 
\cdots \leftrightarrow K_\filtcnt,\]
in which each
$K_i\leftrightarrow K_{i+1}$ is either a forward inclusion $K_i\incto K_{i+1}$
or a backward inclusion $K_i\bakincto K_{i+1}$.
For computational purposes, 
we only consider \emph{cell-wise} filtrations in this paper, i.e.,
each inclusion $K_i\leftrightarrow K_{i+1}$ is an addition or deletion
of a \emph{single} cell;
such an inclusion is sometimes denoted as $K_i\leftrightarrowsp{\sG} K_{i+1}$
with $\sG$ indicating the cell being added or deleted.

We call $\Fcal$ as \emph{non-repetitive} if 
whenever a cell $\sG$ is deleted from $\Fcal$,
the cell $\sG$ is never added again.
We call
$\Fcal$ an {\it up-down} filtration~\cite{carlsson2009zigzag-realvalue} if 
$\Fcal$ can be separated into two parts such that
the first part contains only forward inclusions
and the second part contains only backward ones,
i.e., $\Fcal$ is of the form $\Fcal: K_0 \incto K_1 \incto 
\cdots \incto K_{\ell} \bakincto K_{\ell+1} \bakincto \cdots \bakincto K_\filtcnt$.
Usually in this paper, filtrations start and end with empty complexes,
e.g., $K_0=K_\filtcnt=\emptyset$ in $\Fcal$.

Applying the $\Dim$-th homology functor on $\Fcal$
induces a {\it zigzag module}:
\[\Hm_\Dim(\Fcal): 
\Hm_\Dim(K_0) 
\leftrightarrow
\Hm_\Dim(K_1) 
\leftrightarrow
\cdots 
\leftrightarrow
\Hm_\Dim(K_\filtcnt), \]
in which
each $\Hm_\Dim(K_i)\leftrightarrow \Hm_\Dim(K_{i+1})$
is a linear map induced by inclusion.
It is known~\cite{carlsson2010zigzag,Gabriel72} that
$\Hm_\Dim(\Fcal)$ has a decomposition of the form
$\Hm_\Dim(\Fcal)\simeq\bigoplus_{k\in\LG}\Ical^{[\birth_k,\death_k]}$,
in which each $\Ical^{[\birth_k,\death_k]}$
is a special type of zigzag module called {\it interval module} over the interval $[\birth_k,\death_k]$.
The (multi-)set of intervals
denoted as
$\Pers_\Dim(\Fcal):=\Set{[\birth_k,\death_k]\given k\in\LG}$
is an invariant of $\Fcal$
and is called the {\it $\Dim$-th barcode} of $\Fcal$.
Each interval in $\Pers_\Dim(\Fcal)$ is called a {\it $\Dim$-th persistence interval}
and is also said to be in dimension $\Dim$.
Frequently in this paper, we consider the barcode of $\Fcal$ in all dimensions
$\Pers_*(\Fcal):=\bigsqcup_{\Dim\geq 0}\Pers_\Dim(\Fcal)$.

\begin{definition}[Open and closed birth/death]
\label{dfn:open-close-bd}
For a zigzag filtration
$\Fcal: \emptyset=K_0 \leftrightarrow K_1 \leftrightarrow 
\cdots \leftrightarrow K_\filtcnt=\emptyset$,
the start of any interval in $\Pers_*(\Fcal)$ 
is called a \defemph{birth index}
in $\Fcal$
and the end of any interval is called a \defemph{death index}.
Moreover, 
a birth index $\birth$ is said to be \defemph{closed} 
if $K_{\birth-1}\incto K_\birth$ is a forward inclusion;
otherwise, $\birth$ is  \defemph{open}.
Symmetrically, 
a death index $\death$ is said to be  \defemph{closed}
if $K_{\death}\bakincto K_{\death+1}$ is a backward inclusion;
otherwise, $\death$ is  \defemph{open}.
The types of the birth/death ends
classify intervals in $\Pers_*(\Fcal)$ into four types: 
\defemph{closed-closed}, \defemph{closed-open}, \defemph{open-closed}, and \defemph{open-open}. 
\end{definition}
\begin{remark}
If $\Fcal$ is a levelset zigzag filtration~\cite{carlsson2009zigzag-realvalue},
then the open and closed ends defined above are the same as
for levelset zigzag.
\end{remark}

\begin{remark}\label{rmk:inc-bd}
An inclusion $K_i\leftrightarrow K_{i+1}$ in a cell-wise filtration 
either provides $i+1$ as a birth index
or provides $i$ as a death index (but cannot provide both).
\end{remark}

\paragraph{Mayer-Vietoris diamond.}

The algorithm in this paper
draws upon the Mayer-Vietoris diamond proposed by Carlsson and de Silva~\cite{carlsson2010zigzag}
(see also~\cite{carlsson2019parametrized,carlsson2009zigzag-realvalue}),
which relates barcodes of two filtrations
differing by a local change:

\begin{definition}[Mayer-Vietoris diamond~\cite{carlsson2010zigzag}]
\label{dfn:diamond}
Two cell-wise filtrations $\Fcal$ and $\Fcal'$ 
are related by a \defemph{Mayer-Vietoris diamond}
if they are of the following forms {\rm(}where $\splx\neq\tG${\rm)}:
\begin{equation}
\label{eqn:diamond}
\begin{tikzcd}[column sep=1.6em,
  row sep=0.4em,
]
\Fcal: &[-2em] & & & K_j\arrow[rd,hookleftarrow,pos=0.4,"\tG"]
\\
& K_0\arrow[r,leftrightarrow] & 
  \cdots\arrow[r,leftrightarrow] & 
  K_{j-1}\arrow[ur,hookrightarrow,pos=0.65,"\splx"]
  \arrow[dr,hookleftarrow,pos=0.4,"\tG"] 
  & & 
  K_{j+1}
  \arrow[r,leftrightarrow] & 
  \cdots\arrow[r,leftrightarrow] & 
  K_\filtcnt\\
\Fcal': & & & & K'_j\arrow[ru,hookrightarrow,pos=0.65,"\splx"]
\\
\end{tikzcd}
\end{equation}
In the above diagram,
$\Fcal$ and $\Fcal'$ differ only in the complexes at index $j$
and $\Fcal'$ is derived from $\Fcal$ by switching the addition of $\splx$ 
and deletion of $\tG$.
We also say that $\Fcal'$ is derived from $\Fcal$ by an \defemph{outward} switch %
and $\Fcal$ is derived from $\Fcal'$ by an \defemph{inward} switch.%
\end{definition}

\begin{remark}
In Equation~(\ref{eqn:diamond}),
we only provide a specific form of Mayer-Vietoris diamond
which is sufficient for our purposes;
see~\cite{carlsson2010zigzag,carlsson2009zigzag-realvalue} for a more general form.
According to~\cite{carlsson2010zigzag},
the diamond in Equation~(\ref{eqn:diamond})
is a Mayer-Vietoris diamond 
because $K_j=\linebreak[1]K_{j-1}\union K_{j+1}$
and $K'_j=K_{j-1}\intsec K_{j+1}$.
\end{remark}

We then have the following fact:

\begin{theorem}[Diamond Principle~\cite{carlsson2010zigzag}]
\label{thm:diamond}
Given two cell-wise filtrations $\Fcal,\Fcal'$ 
related by a Mayer-Vietoris diamond
as in Equation~(\ref{eqn:diamond}),
there is a bijection from $\Pers_*(\Fcal)$ to $\Pers_*(\Fcal')$
as follows{\rm:}

\begin{center}
\begin{tabular}{lll}
     \midrule
     $\Pers_*(\Fcal)$ & & $\Pers_*(\Fcal')$ \\
     \midrule
     $[b,j-1]${\rm;} $b\leq j-1$ & $\mapsto$ & $[b,j]$ \\
     $[b,j]${\rm;} $b\leq j-1$ & $\mapsto$ & $[b,j-1]$ \\
     $[j,d]${\rm;} $d\geq j+1$ & $\mapsto$ & $[j+1,d]$ \\
     $[j+1,d]${\rm;} $d\geq j+1$ & $\mapsto$ & $[j,d]$ \\
     $[j,j]$ {\rm of dimension $\Dim$}  & $\mapsto$ & $[j,j]$ {\rm of dimension $\Dim-1$} \\
     $[b,d]${\rm; all other cases} & $\mapsto$ & $[b,d]$ \\
     \midrule
\end{tabular}
\end{center}
Note that
the bijection preserves the dimension of the intervals
except for $[j,j]$.
\end{theorem}
\begin{remark}\label{rmk:diamond}
In the above bijection, only an interval containing {\it some but not all}
of $\Set{j-1,j,j+1}$ maps to a different interval
or different dimension.
\end{remark}

\section{{\sc FastZigzag} algorithm}
\label{sec:fzz}

In this section,
we show that computing 
barcodes for an arbitrary zigzag filtration of simplicial complexes
can be reduced to computing barcodes
for a certain {\it non-zigzag} filtration of $\DG$-complexes.
The resulting algorithm called \textsc{FastZigzag}
is more efficient
considering that standard (non-zigzag) persistence
admits faster algorithms~\cite{bauer2021ripser,bauer2014clear,bauer2017phat,BDM15,chen2011persistent,CK13,zhang2020gpu}
in practice. We confirm the efficiency with experiments in Section~\ref{sec:exp}.

\subsection{Overview}

Given a \emph{simplex-wise} zigzag filtration 
\[\Fcal:
\emptyset=
K_0\leftrightarrowsp{\fsimp{}{0}} K_1\leftrightarrowsp{\fsimp{}{1}}
\cdots 
\leftrightarrowsp{\fsimp{}{\filtcnt-1}} K_\filtcnt
=\emptyset
\]
of simplicial complexes as input,
the \textsc{FastZigzag} algorithm has the following main procedure:
\begin{enumerate}
    \item \label{itm:conv-to-non-rep}
    Convert $\Fcal$ into a {\it non-repetitive} zigzag filtration of $\DG$-complexes.
    \item \label{itm:conv-to-ud}
    Convert the non-repetitive filtration to an \emph{up-down} filtration.
    \item \label{itm:conv-to-std}
    Convert the up-down filtration to a \emph{non-zigzag} filtration 
    with the help of an \emph{extended persistence} filtration.
    \item \label{itm:conv-bc}
    Compute the standard persistence barcode,
    which is then converted to the barcode for the input filtration based on rules
    given in Proposition~\ref{prop:ud-f-map} and~\ref{prop:ef-ud-map}.
\end{enumerate}

Step~\ref{itm:conv-to-non-rep} is achieved by simply treating 
each repeatedly added simplex in $\Fcal$ as a new
cell in the converted filtration (see also~\cite{milosavljevic2011zigzag}).
Throughout the section, we denote the converted
non-repetitive, \emph{cell-wise} filtration as 
\[\Fnr:
\emptyset=
\hat{K}_0\leftrightarrowsp{\fnrsimp{}{0}} \hat{K}_1\leftrightarrowsp{\fnrsimp{}{1}}
\cdots 
\leftrightarrowsp{\fnrsimp{}{\filtcnt-1}} \hat{K}_\filtcnt
=\emptyset.
\]

\begin{figure}
\centering
\includegraphics[width=0.08\textwidth]{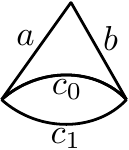}
\caption{The $\DG$-complex resulting from performing an inward switch around $\hat{K}_4$ for the example shown in Figure~\ref{fig:cell-filt}.}
\label{fig:K_hat_4}
\end{figure}

Notice that each $\hat{K}_i$ in $\Fnr$
is homeomorphic to $K_i$ in $\Fcal$,
and hence $\Pers_*(\Fcal)=\Pers_*(\Fnr)$.
However, we get an important difference between $\Fcal$ and $\Fnr$
by treating the simplicial complexes as $\DG$-complexes.
For example, 
in Figure~\ref{fig:cell-filt} presented later in this section,
the first addition of edge $c$ in $\Fcal$ corresponds to a cell $c_0$ in $\Fnr$
and its second addition in $\Fcal$ corresponds to a cell $c_1$.
Performing an inward switch around $\hat{K}_4$ (switching $\bakinctosp{c_0}$ and $\inctosp{c_1}$)
turns $\hat{K}_4$ into a $\DG$-complex as shown in Figure~\ref{fig:K_hat_4}. However, 
we cannot perform such a switch in $\Fcal$ which consists of simplicial complexes,
because diamond switches require the switched simplices or cells 
to be different (see Definition~\ref{dfn:diamond}).

In Section~\ref{sec:conv-to-ud} and~\ref{sec:conv-to-std},
we provide details for Step~\ref{itm:conv-to-ud} and~\ref{itm:conv-to-std}
as well as propositions for converting barcodes
mentioned in Step~\ref{itm:conv-bc}.
We summarize the filtration converting process in Section~\ref{sec:conv-summ}
by providing pseudocodes (Algorithm~\ref{alg:conv-filt}) and examples (Figure~\ref{fig:cell-filt} and~\ref{fig:complete-convert}).

\subsection{Conversion to up-down filtration}
\label{sec:conv-to-ud}

\begin{proposition}\label{prop:norep-filt-UD}
For the filtration $\Fnr$,
there is a cell-wise \emph{up-down} filtration
\[\ud:\emptyset=L_0\hookrightarrow L_1\hookrightarrow
\cdots\hookrightarrow 
L_{\simpcnt}
\hookleftarrow
L_{\simpcnt+1}\hookleftarrow
\cdots
\hookleftarrow L_{2\simpcnt}=\emptyset\]
derived from $\Fnr$ by a sequence of inward switches.
Note that $\filtcnt=2\simpcnt$.
\end{proposition}

\begin{proof}
Let $\hat{K}_{i}\bakinctosp{\fnrsimp{}{i}}\hat{K}_{i+1}$ be the first deletion in $\Fnr$
and $\hat{K}_{j}\inctosp{\fnrsimp{}{j}}\hat{K}_{j+1}$ be the first addition after that.
That is, $\Fnr$ is of the form
\[
\Fnr:\hat{K}_0\hookrightarrow \cdots\hookrightarrow
\hat{K}_{i}\bakinctosp{\fnrsimp{}{i}} 
\hat{K}_{i+1}\bakinctosp{\fnrsimp{}{i+1}}
\cdots
\bakinctosp{\fnrsimp{}{j-2}}
\hat{K}_{j-1}\bakinctosp{\fnrsimp{}{j-1}}
\hat{K}_{j}\inctosp{\fnrsimp{}{j}}\hat{K}_{j+1}
\leftrightarrow\cdots \leftrightarrow \hat{K}_\filtcnt.
\]
Since $\Fnr$ is non-repetitive,
we have $\fnrsimp{}{j-1}\neq\fnrsimp{}{j}$.
So we can switch 
$\bakinctosp{\fnrsimp{}{j-1}}$
and $\inctosp{\fnrsimp{}{j}}$ 
(which is an inward switch) to derive a filtration
\[
\hat{K}_0\hookrightarrow \cdots\hookrightarrow
\hat{K}_{i}\bakinctosp{\fnrsimp{}{i}} 
\hat{K}_{i+1}\bakinctosp{\fnrsimp{}{i+1}}
\cdots
\bakinctosp{\fnrsimp{}{j-2}}
\hat{K}_{j-1}\inctosp{\fnrsimp{}{j}}
\hat{K}'_{j}\bakinctosp{\fnrsimp{}{j-1}}
\hat{K}_{j+1}
\leftrightarrow\cdots \leftrightarrow \hat{K}_\filtcnt.
\]
We then continue performing such inward switches (e.g., the next switch is on 
$\bakinctosp{\fnrsimp{}{j-2}}$
and $\inctosp{\fnrsimp{}{j}}$)
to derive a filtration
\[
\Fnr':\hat{K}_0\hookrightarrow \cdots\hookrightarrow
\hat{K}_{i}\inctosp{\fnrsimp{}{j}}
\hat{K}'_{i+1}\bakinctosp{\fnrsimp{}{i}} 
\cdots
\bakinctosp{\fnrsimp{}{j-3}}
\hat{K}'_{j-1}\bakinctosp{\fnrsimp{}{j-2}}
\hat{K}'_{j}\bakinctosp{\fnrsimp{}{j-1}}
\hat{K}_{j+1}
\leftrightarrow\cdots \leftrightarrow \hat{K}_\filtcnt.
\]
Note that from $\Fnr$ to $\Fnr'$, the up-down `prefix' grows longer.
We can repeat the above operations on the newly derived $\Fnr'$
until the entire filtration turns into an up-down one.
\end{proof}

Throughout the section,
let
\[\ud:\emptyset=\ucplx_0\inctosp{\usimp_0} 
\cdots\inctosp{\usimp_{\simpcnt-1}} 
\ucplx_{\simpcnt}
\bakinctosp{\usimp_{\simpcnt}}
\cdots
\bakinctosp{\usimp_{2\simpcnt-1}} \ucplx_{2\simpcnt}=\emptyset\]
be the up-down filtration for $\Fnr$ 
as described in Proposition~\ref{prop:norep-filt-UD},
where $\filtcnt=2\simpcnt$.
We also let $\hat{K}=\ucplx_{\simpcnt}$.

In a cell-wise filtration,
for a cell $\sG$, let its addition (insertion) be denoted as $\add{\sG}$
and its deletion (removal) be denoted as $\del{\sG}$.
From the proof of Proposition~\ref{prop:norep-filt-UD},
we observe the following: during the transition from $\Fnr$ to $\ud$,
for any two additions $\add{\sG}$ and $\add{\sG'}$ in $\Fnr$ 
(and similarly for deletions),
if $\add{\sG}$ is before $\add{\sG'}$ in $\Fnr$,
then $\add{\sG}$ is also before $\add{\sG'}$ in $\ud$.
We then have the following fact:
\begin{fact}
Given the filtration $\Fnr$, to derive $\ud$, 
one only needs to scan $\Fnr$
and list all the additions first
and then the deletions, following the order in $\Fnr$.
\end{fact}
\begin{remark}
Figure~\ref{fig:F-ud} gives an example of $\Fnr$ and its corresponding $\ud$,
where the additions and deletions in $\Fnr$ and $\ud$
follow the same order.
\end{remark}

\begin{definition}[Creator and destroyer]
\label{dfn:creator-destroyer}
For any interval $[b,d]\in\Pers_*(\Fnr)$, if $\hat{K}_{b-1}\leftrightarrowsp{\fnrsimp_{b-1}}\hat{K}_b$
is forward  {\rm(}resp.\ backward{\rm)}, 
we call $\add{\fnrsimp_{b-1}}$ {\rm(}resp.\ $\del{\fnrsimp_{b-1}}${\rm)} the \defemph{creator}
of $[b,d]$. 
Similarly, if $\hat{K}_{d}\leftrightarrowsp{\fnrsimp_{d}}\hat{K}_{d+1}$
is forward {\rm(}resp.\ backward{\rm)}, 
we call $\add{\fnrsimp_{d}}$ {\rm(}resp.\ $\del{\fnrsimp_{d}}${\rm)} the \defemph{destroyer}
of $[b,d]$.
\end{definition}

By inspecting the interval mapping in the Diamond Principle,
we have the following fact:
\begin{proposition}\label{prop:creat-destr-same}
For two cell-wise filtrations $\Lcal,\Lcal'$ 
related by a Mayer-Vietoris diamond,
any two intervals of $\Pers_*(\Lcal)$ and $\Pers_*(\Lcal')$
mapped by the Diamond Principle
have the same set of creator and destroyer,
though the creator and destroyer may swap.
This observation combined with Proposition~\ref{prop:norep-filt-UD}
implies that there is a bijection from $\Pers_*(\ud)$ to $\Pers_*(\Fnr)$
s.t.\ every two corresponding intervals
have the same set of creator and destroyer.
\end{proposition}
\begin{remark}
The only time when the creator and destroyer swap
in a Mayer-Vietoris diamond
is when the interval $[j,j]$ for the upper filtration in Equation~(\ref{eqn:diamond}) 
turns into the same interval (of one dimension lower) for the lower filtration.
\label{rem:critical}
\end{remark}

Consider the example in Figure~\ref{fig:F-ud}
for an illustration of Proposition~\ref{prop:creat-destr-same}.
In the example,
$[1,2]\in\Pers_1(\Fnr)$ corresponds to $[1,4]\in\Pers_1(\ud)$, 
where their creator is $\add{a}$ and their destroyer is $\del{d}$.
Moreover, $[4,6]\in\Pers_0(\Fnr)$ corresponds to $[4,5]\in\Pers_1(\ud)$.
The creator of $[4,6]\in\Pers_0(\Fnr)$ is $\del{e}$
and the destroyer is $\add{c}$.
Meanwhile, $[4,5]\in\Pers_1(\ud)$ has the same set of creator and destroyer
but the roles swap.

\begin{figure}
  \centering
  \includegraphics[width=\linewidth]{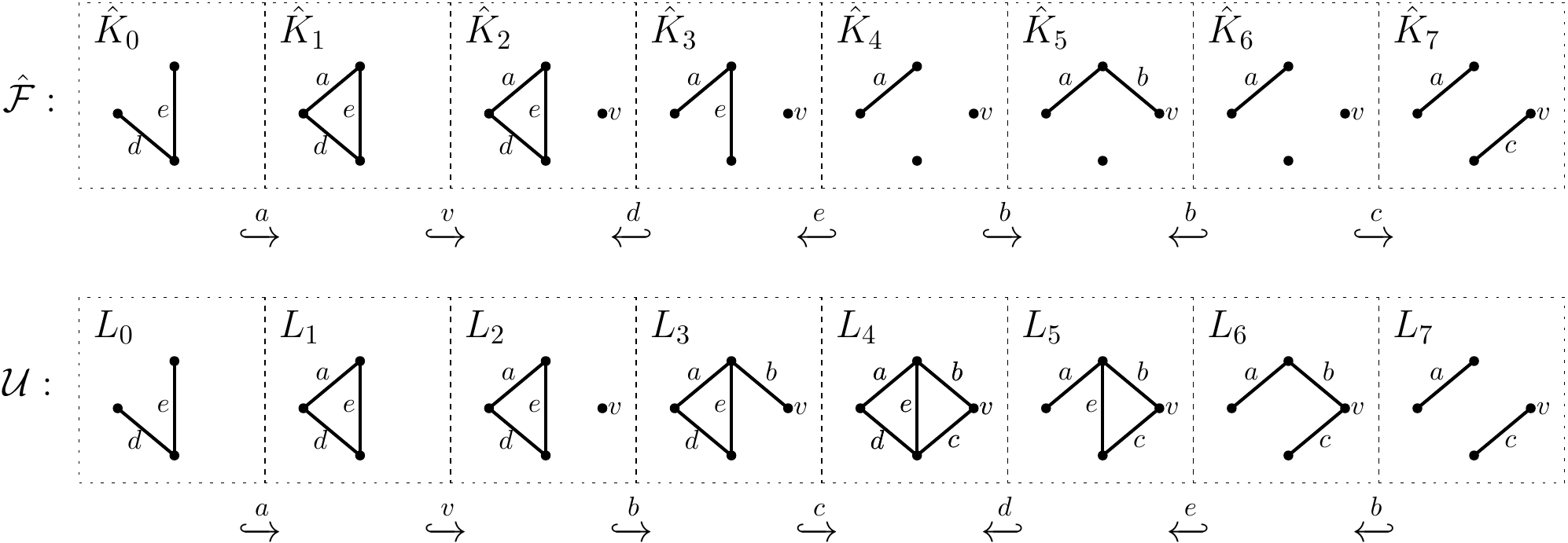}
  \caption{An example of filtration $\Fnr$ and its corresponding up-down filtration $\ud$.
  For brevity, $\Fnr$ does not start and end with empty complexes
  (which can be treated as a truncated case).}
  \label{fig:F-ud}
\end{figure}

For any $\add{\sG}$ or $\del{\sG}$ in $\Fnr$,
let $\id_\Fnr(\add{\sG})$ or $\id_\Fnr(\del{\sG})$ 
denote the index (position) of the addition or deletion.
For example, for an addition $\hat{K}_i\inctosp{\fnrsimp{}{i}}\hat{K}_{i+1}$ in $\Fnr$,
$\id_\Fnr(\add{\fnrsimp{}{i}})=i$.
Proposition~\ref{prop:creat-destr-same} indicates
the following explicit mapping from $\Pers_*(\ud)$ to $\Pers_*(\Fnr)$:

\newcommand{\mytabfont}{\footnotesize}
\newcommand{\mytabfonttwo}{\scriptsize}
\begin{proposition}\label{prop:ud-f-map}
There is a bijection from $\Pers_*(\ud)$ to $\Pers_*(\Fnr)$
which maps each $[b,d]\in\Pers_\Dim(\ud)$
by the following rule{\rm:}
{\rm \begin{center}
\begin{tabular}{cccccl}
\midrule
Type & Condition & & Type & Interval in $\Pers_*(\Fnr)$ & Dim  \\
\midrule
{\mytabfont closed-open} & - & {\mytabfonttwo$\mapsto$} & {\mytabfont closed-open} &
  $\big[\id_\Fnr(\add{\usimp_{b-1}})+1,\id_\Fnr({\add{\usimp_d}})\big]$ &
  $\Dim$ \\ 
\cmidrule{1-6}
{\mytabfont open-closed} & - & {\mytabfonttwo$\mapsto$} & {\mytabfont open-closed} & 
  $\big[\id_\Fnr(\del{\usimp_{b-1}})+1,\id_\Fnr({\del{\usimp_d}})\big]$ &
  $\Dim$ \\ 
\cmidrule{1-6}
\multirow{2}{*}{\mytabfont closed-closed} &
  $\id_\Fnr(\add{\usimp_{b-1}})<\id_\Fnr({\del{\usimp_d}})$ &
  {\mytabfonttwo$\mapsto$} & {\mytabfont closed-closed} &
  $\big[\id_\Fnr(\add{\usimp_{b-1}})+1,\id_\Fnr({\del{\usimp_d}})\big]$ &
  $\Dim$ \\
& $\id_\Fnr(\add{\usimp_{b-1}})>\id_\Fnr({\del{\usimp_d}})$ &
  {\mytabfonttwo$\mapsto$} & {\mytabfont open-open} &
  $\big[\id_\Fnr({\del{\usimp_d}})+1,\id_\Fnr(\add{\usimp_{b-1}})\big]$ &
  $\Dim{-}1$ \\
\midrule
\end{tabular}
\end{center}}
\end{proposition}
\begin{remark}
Notice that $\Pers_*(\ud)$ contains no open-open intervals. 
However, a closed-closed interval $[b,d]\in\Pers_\Dim(\ud)$ turns into an open-open
interval in $\Pers_{\Dim-1}(\Fnr)$ when $\id_\Fnr(\add{\usimp_{b-1}})>\id_\Fnr({\del{\usimp_d}})$.
Such a change
happens 
when a closed-closed interval
turns into a single point interval $[j,j]$ 
during the sequence of outward switches,
after which
the closed-closed interval $[j,j]$ becomes an open-open interval
$[j,j]$ with a dimension shift (see Theorem~\ref{thm:diamond}).
\end{remark}
\begin{remark}
Although it may take $O(\filtcnt^2)$ diamond switches to go from $\Fnr$ to $\ud$
or from $\ud$ to $\Fnr$ as indicated in Proposition~\ref{prop:norep-filt-UD}, 
we observe that these switches 
do not need to be actually executed 
in the algorithm. 
To convert the intervals in $\Pers_*(\ud)$ to those in $\Pers_*(\Fnr)$,
we only need to follow the mapping in Proposition~\ref{prop:ud-f-map},
which takes constant time per interval.
\end{remark}

We can take the example in Figure~\ref{fig:F-ud} for the mapping in
Proposition~\ref{prop:ud-f-map}. 
The interval $[4,5]\in\Pers_1(\ud)$ is a closed-closed one
whose creator is $\add{c}$ and destroyer is $\del{e}$.
We have that $\id_\Fnr(\add{c})=6>\id_\Fnr(\del{e})=3$.
So the corresponding interval in $\Pers_0(\Fnr)$ is
\[
[\id_\Fnr(\del{e})+1,\id_\Fnr(\add{c})]=[4,6].
\]

\subsection{Conversion to non-zigzag filtration}
\label{sec:conv-to-std}

We first convert the up-down filtration $\ud$ to
an extended persistence~\cite{cohen2009extending} filtration $\ef$, which
is then easily converted to an (absolute) non-zigzag filtration using the `coning' technique~\cite{cohen2009extending}.

Inspired by the Mayer-Vietoris pyramid in~\cite{carlsson2009zigzag-realvalue},
we relate $\Pers_*(\ud)$ to the barcode of the filtration $\ef$
 defined as:
\[\ef:\emptyset=\ucplx_0
\incto
\cdots
\incto
\ucplx_{\simpcnt}=(\hat{K},\ucplx_{2\simpcnt})
\incto
(\hat{K},\ucplx_{2\simpcnt-1})
\incto
\cdots
\incto
(\hat{K},\ucplx_{\simpcnt})=(\hat{K},\hat{K})\]
where $\ucplx_{\simpcnt}=\hat{K}=(\hat{K},\ucplx_{2\simpcnt}=\emptyset)$.
When denoting the persistence intervals of $\ef$,
we let the increasing index for the first half of $\ef$ continue to the second half,
i.e., $(\hat{K},L_{2\simpcnt-1})$ has index $\simpcnt+1$
and $(\hat{K},L_{\simpcnt})$ has index $2\simpcnt$.
Then, it can be verified that 
an interval $[b,d]\in\Pers_*(\ef)$ for $b<\simpcnt<d$
starts with the complex $L_b$ and ends with $(\hat{K},L_{3\simpcnt-d})$.

\begin{remark}
A filtration in extended persistence~\cite{cohen2009extending} 
is originally defined for a PL function $f$, where the first half
is the lower-star filtration of $f$
and the second half (the relative part)
is derived from the upper-star filtration of $f$.
The filtration $\ef$ defined above is a generalization of the one in~\cite{cohen2009extending}.
\end{remark}

\begin{proposition}\label{prop:ef-ud-map}
There is a bijection from $\Pers_*(\ef)$ to $\Pers_*(\ud)$
which maps each $[b,d]\in\Pers_*(\ef)$ of dimension $\Dim$
by the following rule{\rm:}
{\rm \begin{center}
\begin{tabular}{llclll}
\midrule
Type & Condition & & Type & Interv. in $\Pers_*(\ud)$ & Dim  \\
\midrule
Ord &
$d<\simpcnt$ & $\mapsto$ & 
  closed-open &
  $[b,d]$ &
  $\Dim$ \\ 
Rel &
$b>\simpcnt$ & $\mapsto$ & 
  open-closed &
  $[3\simpcnt-d,3\simpcnt-b]$ &
  $\Dim{-}1$ \\ 
Ext &
$b\leq\simpcnt\leq d$ & $\mapsto$ & 
  closed-closed &
  $[b,3\simpcnt-d-1]$ &
  $\Dim$ \\
\midrule
\end{tabular}
\end{center}}
\end{proposition}
\begin{remark}
The types `Ord', `Rel', and `Ext' for intervals in $\Pers_*(\ef)$
are as defined in~\cite{cohen2009extending},
which stand for intervals from the {\it ordinary} sub-barcode,
the {\it relative} sub-barcode,
and the {\it extended} sub-barcode.
\end{remark}
\begin{remark}
The above proposition can also be stated by associating the creators and destroyers
as in Proposition~\ref{prop:creat-destr-same} and~\ref{prop:ud-f-map}. The association of additions in the first half of  $\ud$ and $\ef$ 
is straightforward and the deletion of a $\sG$ in $\ud$ is associated with the addition of $\sG$ 
(to the second complex in the pair) in $\ef$.
Then, corresponding intervals in $\Pers_*(\ef)$ and $\Pers_*(\ud)$ in the above proposition 
also have the same set of creators and destroyers.
Combined with Proposition~\ref{prop:creat-destr-same},
we further have that intervals in $\Pers_*(\Fcal)$ and $\Pers_*(\ef)$ can be associated by a
bijection where corresponding intervals have the same pairs of simplices
though they may switch roles of being creators and destroyers. These switches
coincide with the shift in the degree of the homology by having an interval in
$\Pers_p(\ef)$ correspond to an interval in $\Pers_{p-1}(\Fcal)$.
\end{remark}
\begin{proof}
We can build a Mayer-Vietoris pyramid relating the second half 
of $\ef$ and the second half of $\ud$
similar to the one in~\cite{carlsson2009zigzag-realvalue}.
A pyramid for $\simpcnt=4$ is shown in Figure~\ref{fig:pyramid},
where the second half of $\ef$ is along the left side of the triangle
and the second half of $\ud$ is along the bottom.
In Figure~\ref{fig:pyramid}, we represent the second half of $\ef$ and $\ud$
in a slightly different way
considering that $L_4=\hat{K}$ and $L_8=\emptyset$.
Also, each vertical arrow indicates
the addition of a simplex in the second complex of the pair
and each horizontal arrow indicates
the deletion of a simplex in the first complex.

To see the correctness of the mapping,
we first note that
each square in the pyramid is a (more general version of) Mayer-Vietoris diamond as defined in~\cite{carlsson2009zigzag-realvalue}.
Then, 
the mapping stated in the proposition can be verified
using the Diamond Principle (Theorem~\ref{thm:diamond}).
However,
there is a quicker way to verify the mapping
by observing the following: 
corresponding intervals in $\Pers_*(\ef)$ and $\Pers_*(\ud)$
have the same set of creator and destroyer if we
ignore whether it is the addition or deletion of a simplex.
For example, an interval in $\Pers_*(\ef)$ may be created by the addition 
of a simplex $\sG$ in the first half of $\ef$ and 
destroyed by the addition of another simplex $\sG'$ in the second half of $\ef$.
Then, its corresponding interval in $\Pers_*(\ud)$ is 
also created by the addition of $\sG$ in the first half but 
destroyed by the {\it deletion} of $\sG'$ in the second half.
Note that the dimension change for the case $b>\simpcnt$ is caused by
the swap of creator and destroyer.
\end{proof}
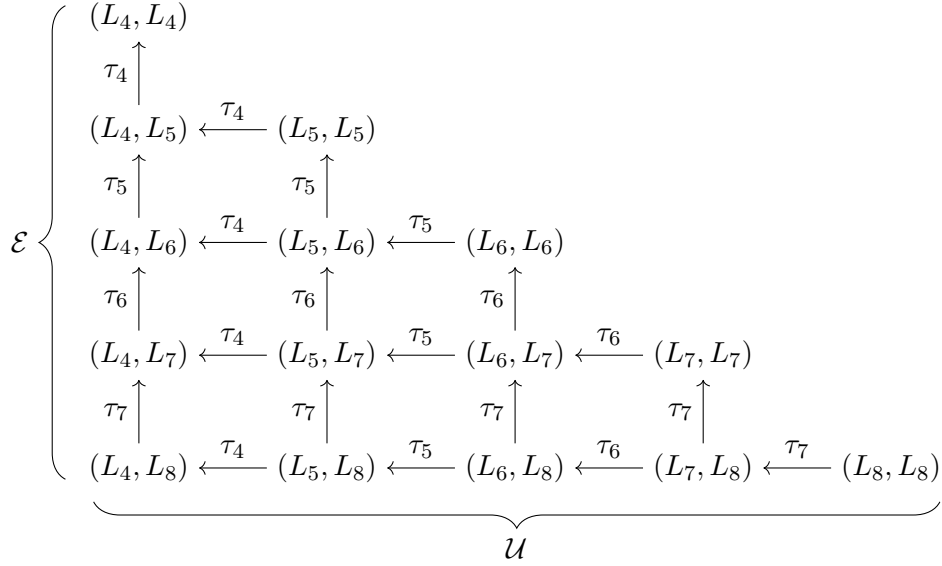
\begin{figure}
\centering
\begin{tikzpicture}[xscale=2.5,yscale=1.5]
\draw (0,0) node(48) {$(L_4,L_8)$} ;
\draw (1,0) node(58) {$(L_5,L_8)$} ;
\draw (2,0) node(68) {$(L_6,L_8)$} ;
\draw (3,0) node(78) {$(L_7,L_8)$} ;
\draw (4,0) node(88) {$(L_8,L_8)$} ;

\draw (0,1) node(47) {$(L_4,L_7)$} ;
\draw (1,1) node(57) {$(L_5,L_7)$} ;
\draw (2,1) node(67) {$(L_6,L_7)$} ;
\draw (3,1) node(77) {$(L_7,L_7)$} ;

\draw (0,2) node(46) {$(L_4,L_6)$} ;
\draw (1,2) node(56) {$(L_5,L_6)$} ;
\draw (2,2) node(66) {$(L_6,L_6)$} ;

\draw (0,3) node(45) {$(L_4,L_5)$} ;
\draw (1,3) node(55) {$(L_5,L_5)$} ;

\draw (0,4) node(44) {$(L_4,L_4)$} ;

\draw[<-] (48) edge node[above]{$\usimp_4$} (58);
\draw[<-] (58) edge node[above]{$\usimp_5$} (68);
\draw[<-] (68) edge node[above]{$\usimp_6$} (78);
\draw[<-] (78) edge node[above]{$\usimp_7$} (88);

\draw[<-] (47) edge node[above]{$\usimp_4$} (57);
\draw[<-] (57) edge node[above]{$\usimp_5$} (67);
\draw[<-] (67) edge node[above]{$\usimp_6$} (77);

\draw[<-] (46) edge node[above]{$\usimp_4$} (56);
\draw[<-] (56) edge node[above]{$\usimp_5$} (66);

\draw[<-] (45) edge node[above]{$\usimp_4$} (55);

\draw[->] (48) edge node[left]{$\usimp_7$} (47);
\draw[->] (47) edge node[left]{$\usimp_6$} (46);
\draw[->] (46) edge node[left]{$\usimp_5$} (45);
\draw[->] (45) edge node[left]{$\usimp_4$} (44);

\draw[->] (58) edge node[left]{$\usimp_7$} (57);
\draw[->] (57) edge node[left]{$\usimp_6$} (56);
\draw[->] (56) edge node[left]{$\usimp_5$} (55);

\draw[->] (68) edge node[left]{$\usimp_7$} (67);
\draw[->] (67) edge node[left]{$\usimp_6$} (66);

\draw[->] (78) edge node[left]{$\usimp_7$} (77);

\draw [decorate,decoration={brace,amplitude=10pt},xshift=-4pt,yshift=0pt]
(-0.25,-0.1) --
node [midway,above,sloped,yshift=10pt] {\rotatebox{-90}{$\ef$}}
(-0.25,4.1);

\draw [decorate,decoration={brace,amplitude=10pt,mirror},xshift=-4pt,yshift=0pt]
(-0.1,-0.3) -- 
node [midway,below,yshift=-10pt] {$\ud$}
(4.4,-0.3);
\end{tikzpicture}
\caption{A Mayer-Vietoris pyramid relating the second half 
of $\ef$ and $\ud$ for $\simpcnt=4$.}
\label{fig:pyramid}
\end{figure}

By Proposition~\ref{prop:ud-f-map} and~\ref{prop:ef-ud-map}, we
only need to compute $\Pers_*(\ef)$
in order to compute $\Pers_*(\Fcal)$.
The barcode of $\ef$ can be computed using the `coning' technique~\cite{cohen2009extending},
which converts $\ef$ into an (absolute) non-zigzag filtration $\hat{\ef}$.
Specifically, let $\oG$ be a vertex different from all vertices in $\hat{K}$.
For a $\Dim$-cell $\sG$ of $\hat{K}$,
we let $\oG\cdot\sG$ denote the {\it cone} of $\sG$,
which is a $(\Dim+1)$-cell having 
cells $\Set{\sG}\union\Set{\oG\cdot\tG\mid\tG\in\partial\sG}$
in its boundary.
The \emph{cone} $\oG\cdot L_i$ of a complex $L_i$ consists of three parts: 
the vertex $\oG$, $L_i$, and cones of all cells of $L_i$.
The filtration $\hat{\ef}$ is then defined as~\cite{cohen2009extending}:
\[\hat{\ef}:
\ucplx_0\union \Set{\oG}
\incto
\cdots
\incto
\ucplx_{\simpcnt}\union \Set{\oG}
=\hat{K}\union\oG\cdot\ucplx_{2\simpcnt}
\incto
\hat{K}\union\oG\cdot\ucplx_{2\simpcnt-1}
\incto
\cdots
\incto
\hat{K}\union\oG\cdot\ucplx_{\simpcnt}
\]
We have that $\Pers_*(\ef)$ equals $\Pers_*(\hat{\ef})$
discarding the only infinite interval~\cite{cohen2009extending}.
Note that if a cell $\sG$ is added (to the second complex) 
from $(\hat{K},\ucplx_{i})$ to $(\hat{K},\ucplx_{i-1})$ in $\ef$,
then the cone $\oG\cdot\sG$ is added 
from $\hat{K}\union\oG\cdot\ucplx_{i}$ to $\hat{K}\union\oG\cdot\ucplx_{i-1}$ in $\hat{\ef}$.

\subsection{Summary of filtration conversion}
\label{sec:conv-summ}

We summarize the filtration conversion process described in this section
in Algorithm~\ref{alg:conv-filt},
in which we assume that each simplex in $\Fcal$ is given
by its set of vertices.
The converted standard filtration $\hat{\ef}$
is represented 
by its boundary matrix $D$, whose columns (and equivalently the chains they represent)
are treated as sets of identifiers of the boundary cells.
Algorithm~\ref{alg:conv-filt} also maintains the following data structures:
\begin{itemize}
    \item 
    $\mathtt{cid}$ denotes the map from a simplex $\sG$
    to the identifier of the \emph{most recent} copy of cell corresponding to $\sG$.
    \item 
    $\mathtt{del\_list}$ denotes the list of cell identifiers deleted in the input filtration.
    \item 
    $\mathtt{cone\_id}$ denotes the map from the identifier of a cell
    to that of its coned cell.
\end{itemize}

\begin{algorithm}[th!]
\caption{Pseudocode for converting input filtration}

\vspace{1ex}

\begin{algorithmic}[1]
\Procedure{ConvertFilt}{$\Fcal$}
\State initialize boundary matrix $D$, cell-id map $\mathtt{cid}$, deleted cell list $\mathtt{del\_list}$ as empty
\State append an empty column to $D$ representing vertex $\oG$ for coning
\vspace{0.3em}
\State $\mathtt{id}\leftarrow 1$
\Comment{variable keeping track of id for cells}
\ForEach{$K_i\leftrightarrowsp{\sG_i}K_{i+1}$ in $\Fcal$}
\If{$\sG_i$ is being inserted}
\State $\mathtt{cid}[\sG_i]=\mathtt{id}$ 
\Comment{get a new cell as a copy of simplex $\sG_i$}
\State $\mathtt{col}\leftarrow \textsc{CellBoundary}(\sG_i,\mathtt{cid})$
\label{line:CellBoundary}
\State append $\mathtt{col}$ to $D$
\State $\mathtt{id}\leftarrow\mathtt{id}+1$
\Else
\State append $\mathtt{cid}[\sG_i]$ to $\mathtt{del\_list}$
\EndIf
\EndFor
\State initialize map $\mathtt{cone\_id}$ as empty \Comment{$\mathtt{cone\_id}$ tracks id for coned cells}
\ForEach{$\mathtt{del\_id}$ in $\mathtt{del\_list}$ (accessed reversely)}
\State $\mathtt{cone\_id}[\mathtt{del\_id}]\leftarrow\mathtt{id}$
\Comment{get a new coned cell}
\State $\mathtt{col}\leftarrow \textsc{ConedCellBoundary}(\mathtt{del\_id},D,\mathtt{cone\_id})$
\label{line:ConedCellBoundary}
\State append $\mathtt{col}$ to $D$
\State $\mathtt{id}\leftarrow\mathtt{id}+1$
\EndFor

\State \Return $D$ 
\EndProcedure
\end{algorithmic}
\label{alg:conv-filt}
\end{algorithm}

Subroutine \textsc{CellBoundary} in Line~\ref{line:CellBoundary}
converts boundary simplices of $\sG_i$
to a column of cell identifiers
based on the map $\mathtt{cid}$.
Subroutine \textsc{ConedCellBoundary} in Line~\ref{line:ConedCellBoundary}
returns boundary column for the cone of the cell identified by $\mathtt{del\_id}$.

\begin{figure}[!tbh]
  \centering
  \includegraphics[width=\linewidth]{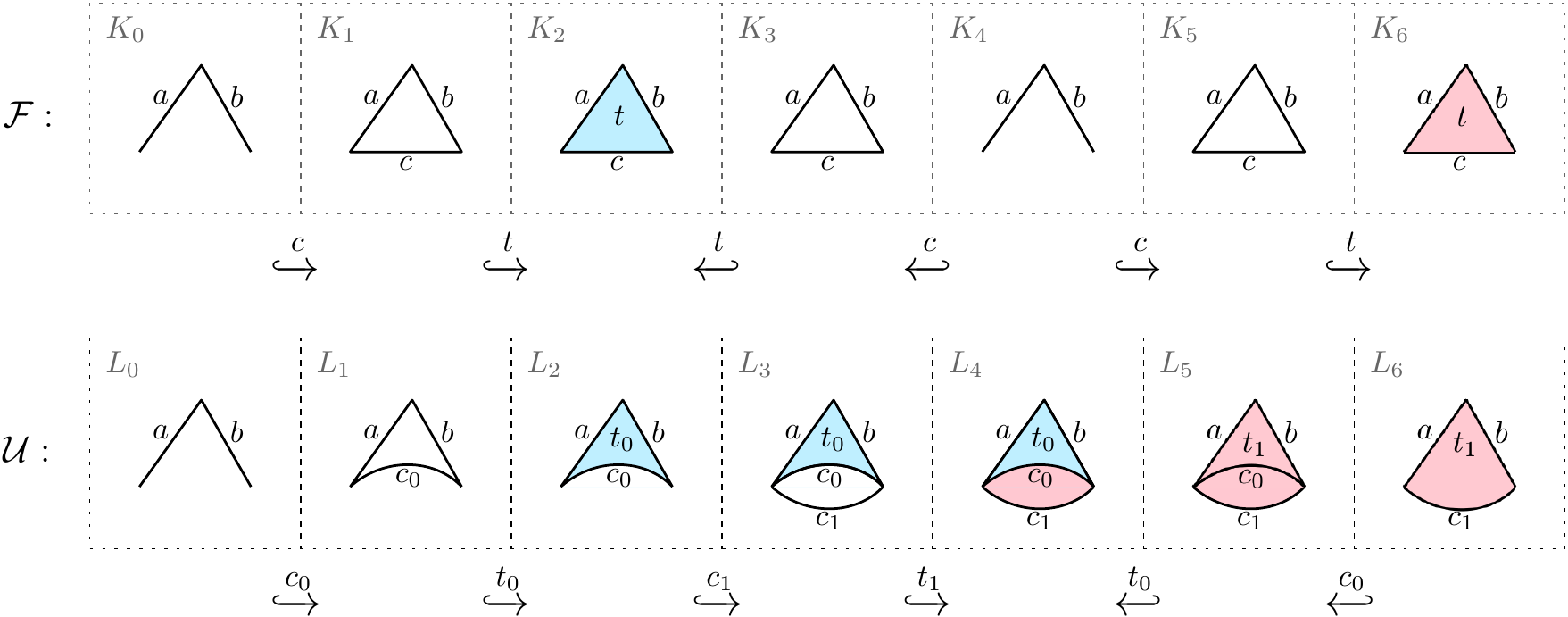}
  \caption{An example of an up-down cell-wise filtration $\ud$ built
from a given simplex-wise filtration $\Fcal$. 
For brevity, $\Fcal$ does not start and end with empty complexes. The final conversion
to $\hat{\ef}$ is not shown for this example due to page-width constraint. A complete
conversion for a smaller example is shown in Figure~\ref{fig:complete-convert}.}
  \label{fig:cell-filt}
\end{figure}

\begin{figure}[!tbh]
  \centering
  \includegraphics[width=\linewidth]{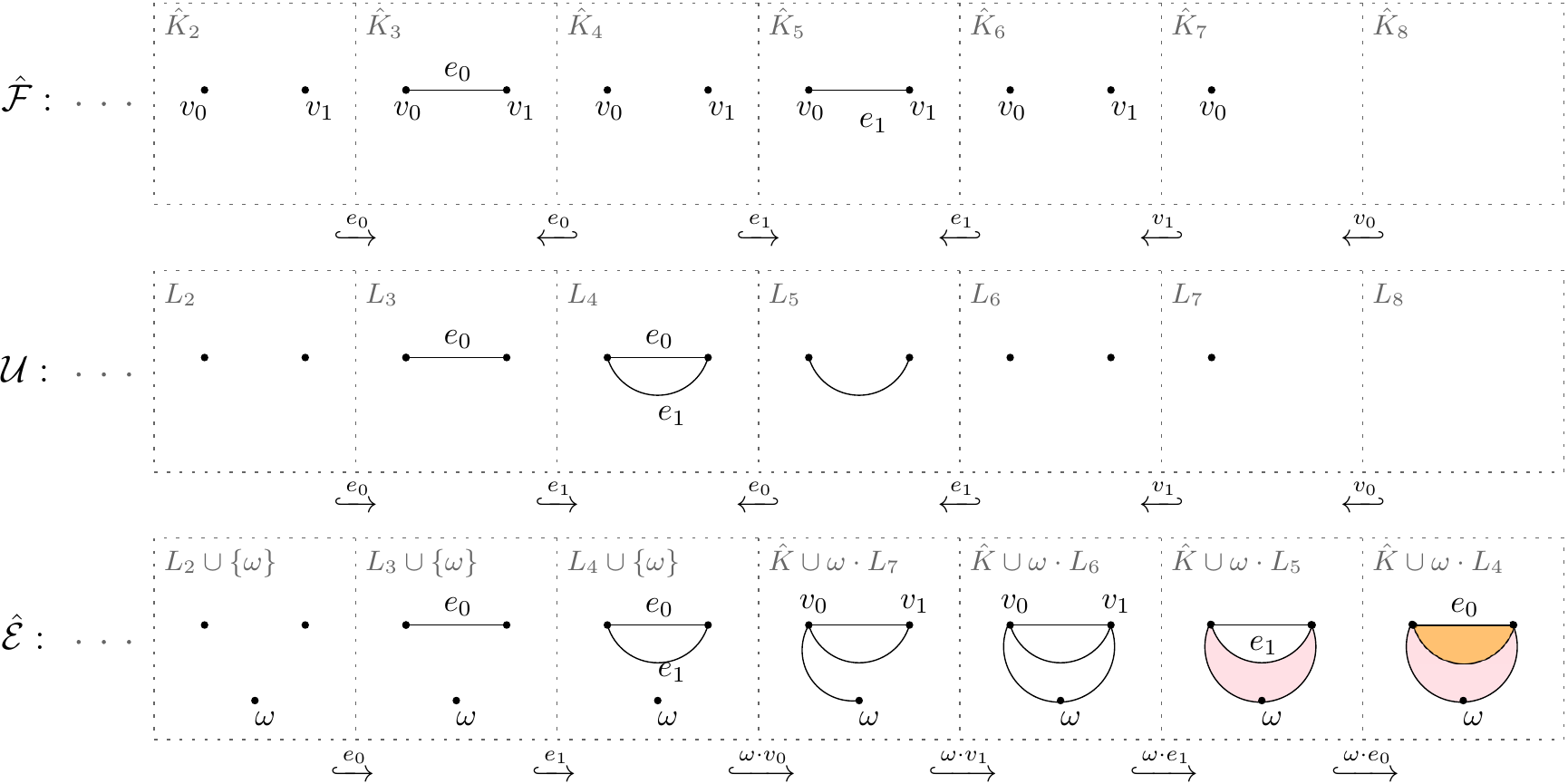}
  \caption{An example of converting a zigzag filtration $\hat{\Fcal}$ to a non-zigzag filtration.}
  \label{fig:complete-convert}
\end{figure}

We provide an example of the up-down cell-wise filtration $\ud$ built
from a given simplex-wise filtration $\Fcal$ in Figure~\ref{fig:cell-filt}.
In the example, edge $c$ and triangle $t$ are repeatedly added twice in $\Fcal$,
and therefore each correspond to two copies of cells in $\ud$. We provide 
another example of a complete conversion from a given zigzag filtration
to a non-zigzag filtration in Figure~\ref{fig:complete-convert}.

With the \textsc{ConvertFilt} subroutine,
Algorithm~\ref{alg:fzz} provides a concise summary of \textsc{FastZigzag}.
Given that for a filtration $\Fcal$ of length $\filtcnt$, \textsc{ConvertFilt}
takes $O(m)$ time
and  \textsc{ConvertBarcode} takes $O(1)$ time per bar,
we now have the following conclusion:
\begin{algorithm}[h!]
\caption{Pseudocode for \textsc{FastZigzag}}

\vspace{1ex}

\begin{algorithmic}[1]
\Procedure{FastZigzag}{$\Fcal$}
\State $D\leftarrow \textsc{ConvertFilt}(\Fcal)$
\State $B\leftarrow \textsc{ComputeBarcode}(D)$
\State $B'\leftarrow \textsc{ConvertBarcode}(B)$
\State \Return $B'$ 
\EndProcedure
\end{algorithmic}
\label{alg:fzz}
\end{algorithm}

\begin{theorem}
Given a simplex-wise zigzag filtration $\Fcal$
with length $\filtcnt$, {\sc FastZigzag} computes $\Pers_*(\Fcal)$ in time $T(\filtcnt)+O(\filtcnt)$,
where $T(\filtcnt)$ is the time used for computing the barcode
of a non-zigzag cell-wise filtration with length $\filtcnt$.
\end{theorem}
\begin{remark}
Theoretically, $T(m)=O(m^\oG)$~\cite{milosavljevic2011zigzag}, 
where $\omega< 2.37286$ is the matrix multiplication exponent~\cite{alman2021refined}.
So the theoretical complexity of \textsc{FastZigzag}
is $O(m^\oG)$.
\end{remark}

\subsection{Experiments}
\label{sec:exp}
We implement the \textsc{FastZigzag} algorithm 
described in this section and compare the performance
with \Dionysustwo{}~\cite{Dionysus2} 
(implementing the algorithm in~\cite{carlsson2009zigzag-realvalue})
and \Gudhi{}\footnote{The code is shared by personal communication.}~\cite{gudhi:urm}
(implementing the algorithm in~\cite{maria2014zigzag,maria2019discrete}).
When implementing \textsc{FastZigzag}, we utilize the \Phat{}~\cite{bauer2017phat} 
software for computing non-zigzag persistence.
Our implementation is publicly available through: 
\url{https://github.com/taohou01/fzz}.

To test the performance,
we generate eleven
simplex-wise filtrations of similar lengths (5$\sim$6 millions; see Table~\ref{tab:perform}).
The reason for using filtrations of similar lengths
is to test the impact of \emph{repetitiveness} on the performance
for different algorithms, where repetitiveness is the average times
a simplex is repeatedly added in a filtration 
(e.g., repetitiveness being 1 means that
the filtration is non-repetitive).
We utilize three different approaches for generating the filtrations:
\begin{itemize}
    \item
The two non-repetitive filtrations (No.\ 1 and 2) 
are generated by first taking a simplicial complex with vertices in $\Real^3$,
and then taking the height function $h$ along a certain axis.
After this, we build an up-down filtration for the complex 
where the first half is the lower-star filtration of $h$ %
and the second half is the (reversed) upper-star filtration of $h$. %
We then randomly perform outward switches on the up-down filtration
to derive a non-repetitive filtration.
Note that the simplicial complex is derived from a triangular mesh 
supplemented by a Vietoris-Rips complex on the vertices;
one triangular mesh (Dragon) is downloaded from 
the Stanford Computer Graphics Laboratory.

    \item
Filtration No.\ 3 -- 8 are generated from a sequence of edge additions
and deletions, for which we then take the \emph{clique complex} 
(up to a certain dimension)
for each edge set in the sequence.
The edge sequence is derived by randomly adding and deleting edges
for a set of points.

    \item 
The remaining filtrations (No.\ 9 -- 11) are the \emph{oscillating Rips zigzag}~\cite{oudot2015zigzag}
generated from point clouds of size 2000 -- 4000 sampled from some triangular meshes
(Space Shuttle from an online repository\footnote{Ryan Holmes: \url{http://www.holmes3d.net/graphics/offfiles/}};
Bunny and Dragon from the Stanford Computer Graphics Laboratory).
\end{itemize}

Table~\ref{tab:perform} lists running time of the 
three algorithms on all filtrations,
where the length, maximum dimension (D),
repetitiveness (Rep), and maximum complex size (MaxK) 
are also provided for each filtration.
From Table~\ref{tab:perform}, we observe that \textsc{FastZigzag} ($\text{T}_\textsc{FZZ}$) consistently
achieves the best running time across all inputs, 
with significant speedups (see column `SU' in Table~\ref{tab:perform}).
The speedup is calculated as the min-time of \Dionysustwo{} and \Gudhi{}
divided by the time of \textsc{FastZigzag}.
Notice that since \Gudhi{} only takes a sequence of edge additions
and deletions as input (and builds clique complexes on-the-fly),
we do not run \Gudhi{} on the first two inputs in Table~\ref{tab:perform}, 
which are only given as simplex-wise filtrations.
We also observe that the speedup of \textsc{FastZigzag} tends to be less prominent
as the repetitiveness increases.
This is because higher repetitiveness leads to smaller max/average complex size
in the input zigzag filtration,
so that algorithms directly working on the input filtration
could have less processing time~\cite{carlsson2009zigzag-realvalue,maria2014zigzag,maria2019discrete}.
On the other hand, the complex size
in the converted non-zigzag filtration that \textsc{FastZigzag} works on is always increasing.

\newcommand{\tabnum}[1]{\tt\small #1}
\begin{table}[!htb]
\centering
\caption{Running time of \Dionysustwo{}, \Gudhi{}, and \textsc{FastZigzag}
on different filtrations of similar lengths with various repetitiveness.
All tests were run on a desktop with Intel(R) Core(TM) i5-9500 CPU @ 3.00GHz, 16GB memory, and Linux OS.}
\label{tab:perform}
\begin{tabular}{cccrrrrrr}
\midrule
No.\  & {Length} & D & Rep & MaxK & {$\text{T}_\textsc{Dio2}$} & {$\text{T}_\textsc{Gudhi}$} & {$\text{T}_\textsc{FZZ}$} & SU\\
\midrule
1 & \tabnum{5,260,700} & \tabnum{5} & \tabnum{1.0} & 
  \tabnum{883,350} & 
  \tabnum{2h02m46.0s} & 
  \tabnum{$-$} & 
  \tabnum{8.9s} &
  \tabnum{873}
\\
2 & \tabnum{5,254,620} & \tabnum{4} & \tabnum{1.0} & 
  \tabnum{1,570,326} & 
  \tabnum{19m36.6s} & 
  \tabnum{$-$} & 
  \tabnum{11.0s} &
  \tabnum{107}
\\
3 & \tabnum{5,539,494} & \tabnum{5} & \tabnum{1.3} & 
  \tabnum{1,671,047} & 
  \tabnum{3h05m00.0s} & 
  \tabnum{45m47.0s} & 
  \tabnum{3m20.8s} &
  \tabnum{13.7}
\\
4 & \tabnum{5,660,248} & \tabnum{4} & \tabnum{2.0} & 
  \tabnum{1,385,979} & 
  \tabnum{2h59m57.0s} & 
  \tabnum{29m46.7s} & 
  \tabnum{4m59.5s} &
  \tabnum{6.0}
\\
5 & \tabnum{5,327,422} & \tabnum{4} & \tabnum{3.5} & 
  \tabnum{760,098} & 
  \tabnum{43m54.8s} & 
  \tabnum{10m35.2s} & 
  \tabnum{3m32.1s} &
  \tabnum{3.0}
\\
6 & \tabnum{5,309,918} & \tabnum{3} & \tabnum{5.1} & 
  \tabnum{523,685} & 
  \tabnum{5h46m03.0s} & 
  \tabnum{1h32m37.0s} & 
  \tabnum{19m30.2s} &
  \tabnum{4.7}
\\
7 & \tabnum{5,357,346} & \tabnum{3} & \tabnum{7.3} & 
  \tabnum{368,830} & 
  \tabnum{3h37m54.0s} & 
  \tabnum{57m28.4s} & 
  \tabnum{30m25.2s} &
  \tabnum{1.9}
\\
8 & \tabnum{6,058,860} & \tabnum{4} & \tabnum{9.1} & 
  \tabnum{331,211} & 
  \tabnum{53m21.2s} & 
  \tabnum{7m19.0s} & 
  \tabnum{3m44.4s} &
  \tabnum{2.0}
\\
9 & \tabnum{5,135,720} & \tabnum{3} & \tabnum{21.9} & 
  \tabnum{11,859} & 
  \tabnum{23.8s} &
  \tabnum{15.6s} &
  \tabnum{8.6s} &
  \tabnum{1.9} 
\\
10 & \tabnum{5,110,976} & \tabnum{3} & \tabnum{27.7} & 
  \tabnum{11,435} & 
  \tabnum{36.2s} &
  \tabnum{39.9s} &
  \tabnum{8.5s} &
  \tabnum{4.3} 
\\
11 & \tabnum{5,811,310} & \tabnum{4} & \tabnum{44.2} & 
  \tabnum{7,782} & 
  \tabnum{38.5s} &
  \tabnum{36.9s} &
  \tabnum{23.9s} &
  \tabnum{1.5} 
\\
\midrule
\end{tabular}
\end{table}

Table~\ref{tab:mem} lists the memory consumption of the three algorithms. We observe that 
\textsc{FastZigzag} tends to consume more memory than the other two 
on the non-repetitive filtrations (No.\ 1 and 2) and the random clique filtrations (No.\ 3 -- 8).
However, \textsc{FastZigzag} is consistently achieving the best memory footprint 
on the oscillating Rips filtrations (No.\ 9 -- 11) despite the high repetitiveness.

\begin{table}[!htb]
\centering
\caption{Memory consumption (in gigabytes) of the three algorithms on all filtrations.}
\label{tab:mem}
\begin{tabular}{ccrrrrr}
\midrule
No.\  & {Length} & Rep & MaxK & {$\text{M}_\textsc{Dio2}$} & {$\text{M}_\textsc{Gudhi}$} & {$\text{M}_\textsc{FZZ}$}\\
\midrule
1 & \tabnum{5,260,700} & \tabnum{1.0} & 
  \tabnum{883,350} & 
  \tabnum{3.23} & 
  \tabnum{$-$} & 
  \tabnum{0.59} 
\\
2 & \tabnum{5,254,620} & \tabnum{1.0} & 
  \tabnum{1,570,326} & 
  \tabnum{3.93} & 
  \tabnum{$-$} & 
  \tabnum{0.61} 
\\
3 & \tabnum{5,539,494} & \tabnum{1.3} & 
  \tabnum{1,671,047} & 
  \tabnum{15.52} & 
  \tabnum{13.49} &
  \tabnum{9.76}
\\
4 & \tabnum{5,660,248} & \tabnum{2.0} & 
  \tabnum{1,385,979} & 
  \tabnum{7.64} & 
  \tabnum{8.43} & 
  \tabnum{11.04} 
\\
5 & \tabnum{5,327,422} & \tabnum{3.5} & 
  \tabnum{760,098} & 
  \tabnum{3.27} & 
  \tabnum{3.40} & 
  \tabnum{6.22} 
\\
6 & \tabnum{5,309,918} & \tabnum{5.1} & 
  \tabnum{523,685} & 
  \tabnum{4.94} & 
  \tabnum{5.27} & 
  \tabnum{10.23} 
\\
7 & \tabnum{5,357,346} & \tabnum{7.3} & 
  \tabnum{368,830} & 
  \tabnum{4.03} & 
  \tabnum{3.91} & 
  \tabnum{8.19}
\\
8 & \tabnum{6,058,860} & \tabnum{9.1} & 
  \tabnum{331,211} & 
  \tabnum{2.12} & 
  \tabnum{1.48} & 
  \tabnum{3.68} 
\\
9 & \tabnum{5,135,720} & \tabnum{21.9} & 
  \tabnum{11,859} & 
  \tabnum{0.92} &
  \tabnum{0.47} &
  \tabnum{0.50} 
\\
10 & \tabnum{5,110,976} & \tabnum{27.7} & 
  \tabnum{11,435} & 
  \tabnum{0.88} &
  \tabnum{0.48} &
  \tabnum{0.47} 
\\
11 & \tabnum{5,811,310} & \tabnum{44.2} & 
  \tabnum{7,782} & 
  \tabnum{0.95} &
  \tabnum{0.60} &
  \tabnum{0.51} 
\\
\midrule
\end{tabular}
\end{table}

\section{Conclusions}

In this paper, we propose a zigzag persistence algorithm 
called \textsc{FastZigzag}
by first treating repeatedly added simplices in an input zigzag filtration
as distinct copies and then converting the input filtration
to a non-zigzag filtration.
The barcode of the converted non-zigzag filtration can then be
easily mapped back to barcode of the input zigzag filtration. 
The efficiency of our algorithm is confirmed by experiments.
This research also brings forth the following open questions:
\begin{itemize}
    \item Parallel versions~\cite{carlsson2019persistent,zhang2020gpu} of the algorithms for computing 
    standard and zigzag exist.
    While the computation of standard persistence in our \textsc{FastZigzag} algorithm 
    can directly utilize the existing parallelization techniques, we ask
    if the conversions done in \textsc{FastZigzag} can be efficiently parallelized.
    Such an extension can provide further speedups by harnessing multi-cores.
    
    \item While persistence intervals are important topological descriptors,
    their \emph{representatives} also reveal critical information
    (e.g., a recently proposed algorithm~\cite{dey2021updating} for updating
    zigzag barcodes over local changes uses
    representatives explicitly).
    Can the \textsc{FastZigzag} algorithm be adapted so that representatives
    for the input zigzag filtration are recovered from representatives 
    for the converted non-zigzag filtration?
\end{itemize}

\section*{Acknowledgment:}
We thank the Stanford Computer Graphics Laboratory
and Ryan Holmes
for providing the triangular meshes used in the experiment
of this paper.

\bibliographystyle{plainurl}
\bibliography{refs}

\begin{thebibliography}{10}

\bibitem{alman2021refined}
Josh Alman and Virginia~Vassilevska Williams.
\newblock A refined laser method and faster matrix multiplication.
\newblock In {\em Proceedings of the 2021 ACM-SIAM Symposium on Discrete
  Algorithms (SODA)}, pages 522--539. SIAM, 2021.

\bibitem{bauer2021ripser}
Ulrich Bauer.
\newblock Ripser: Efficient computation of vietoris--rips persistence barcodes.
\newblock {\em Journal of Applied and Computational Topology}, 5(3):391--423,
  2021.

\bibitem{bauer2014clear}
Ulrich Bauer, Michael Kerber, and Jan Reininghaus.
\newblock Clear and compress: Computing persistent homology in chunks.
\newblock In {\em Topological methods in data analysis and visualization III},
  pages 103--117. Springer, 2014.

\bibitem{bauer2017phat}
Ulrich Bauer, Michael Kerber, Jan Reininghaus, and Hubert Wagner.
\newblock Phat -- persistent homology algorithms toolbox.
\newblock {\em Journal of Symbolic Computation}, 78:76--90, 2017.

\bibitem{BDM15}
Jean{-}Daniel Boissonnat, Tamal~K. Dey, and Cl{\'{e}}ment Maria.
\newblock The compressed annotation matrix: An efficient data structure for
  computing persistent cohomology.
\newblock {\em Algorithmica}, 73(3):607--619, 2015.

\bibitem{carlsson2010zigzag}
Gunnar Carlsson and Vin de~Silva.
\newblock Zigzag persistence.
\newblock {\em Foundations of Computational Mathematics}, 10(4):367--405, 2010.

\bibitem{carlsson2019parametrized}
Gunnar Carlsson, Vin de~Silva, Sara Kali{\v{s}}nik, and Dmitriy Morozov.
\newblock Parametrized homology via zigzag persistence.
\newblock {\em Algebraic \& Geometric Topology}, 19(2):657--700, 2019.

\bibitem{carlsson2009zigzag-realvalue}
Gunnar Carlsson, Vin de~Silva, and Dmitriy Morozov.
\newblock Zigzag persistent homology and real-valued functions.
\newblock In {\em Proceedings of the Twenty-Fifth Annual Symposium on
  Computational Geometry}, pages 247--256, 2009.

\bibitem{carlsson2019persistent}
Gunnar Carlsson, Anjan Dwaraknath, and Bradley~J. Nelson.
\newblock Persistent and zigzag homology: A matrix factorization viewpoint.
\newblock {\em arXiv preprint arXiv:1911.10693}, 2019.

\bibitem{chen2011persistent}
Chao Chen and Michael Kerber.
\newblock Persistent homology computation with a twist.
\newblock In {\em Proceedings 27th European Workshop on Computational
  Geometry}, volume~11, pages 197--200, 2011.

\bibitem{CK13}
Chao Chen and Michael Kerber.
\newblock An output-sensitive algorithm for persistent homology.
\newblock {\em Comput. Geom.: Theory and Applications}, 46(4):435--447, 2013.

\bibitem{cohen2009extending}
David Cohen-Steiner, Herbert Edelsbrunner, and John Harer.
\newblock Extending persistence using {P}oincar{\'e} and {L}efschetz duality.
\newblock {\em Foundations of Computational Mathematics}, 9(1):79--103, 2009.

\bibitem{de2011dualities}
Vin de~Silva, Dmitriy Morozov, and Mikael Vejdemo-Johansson.
\newblock Dualities in persistent (co)homology.
\newblock {\em Inverse Problems}, 27(12):124003, 2011.

\bibitem{derksen2005quiver}
Harm Derksen and Jerzy Weyman.
\newblock Quiver representations.
\newblock {\em Notices of the AMS}, 52(2):200--206, 2005.

\bibitem{dey2021graph}
Tamal~K. Dey and Tao Hou.
\newblock Computing zigzag persistence on graphs in near-linear time.
\newblock In {\em 37th International Symposium on Computational Geometry, SoCG
  2021}, volume 189 of {\em LIPIcs}, pages 30:1--30:15. Schloss Dagstuhl -
  Leibniz-Zentrum f{\"{u}}r Informatik, 2021.

\bibitem{dey2021updating}
Tamal~K. Dey and Tao Hou.
\newblock Updating zigzag persistence and maintaining representatives over
  changing filtrations.
\newblock {\em arXiv preprint arXiv:2112.02352}, 2021.

\bibitem{dey2021computing}
Tamal~K. Dey, Woojin Kim, and Facundo M{\'{e}}moli.
\newblock Computing generalized rank invariant for 2-parameter persistence
  modules via zigzag persistence and its applications.
\newblock In {\em 38th International Symposium on Computational Geometry, SoCG
  2022}, volume 224 of {\em LIPIcs}, pages 34:1--34:17, 2022.

\bibitem{edelsbrunner2000topological}
Herbert Edelsbrunner, David Letscher, and Afra Zomorodian.
\newblock Topological persistence and simplification.
\newblock In {\em Proceedings 41st Annual Symposium on Foundations of Computer
  Science}, pages 454--463. IEEE, 2000.

\bibitem{Gabriel72}
Peter Gabriel.
\newblock {Unzerlegbare Darstellungen I}.
\newblock {\em Manuscripta Mathematica}, 6(1):71--103, 1972.

\bibitem{hatcher2002algebraic}
Allen Hatcher.
\newblock {\em Algebraic Topology}.
\newblock Cambridge University Press, 2002.

\bibitem{holme2012temporal}
Petter Holme and Jari Saram{\"a}ki.
\newblock Temporal networks.
\newblock {\em Physics Reports}, 519(3):97--125, 2012.

\bibitem{maria2014zigzag}
Cl{\'e}ment Maria and Steve~Y. Oudot.
\newblock Zigzag persistence via reflections and transpositions.
\newblock In {\em Proceedings of the Twenty-Sixth Annual ACM-SIAM Symposium on
  Discrete Algorithms}, pages 181--199. SIAM, 2014.

\bibitem{maria2016computing}
Cl{\'e}ment Maria and Steve~Y. Oudot.
\newblock Computing zigzag persistent cohomology.
\newblock {\em arXiv preprint arXiv:1608.06039}, 2016.

\bibitem{maria2019discrete}
Cl{\'e}ment Maria and Hannah Schreiber.
\newblock Discrete morse theory for computing zigzag persistence.
\newblock In {\em Workshop on Algorithms and Data Structures}, pages 538--552.
  Springer, 2019.

\bibitem{milosavljevic2011zigzag}
Nikola Milosavljevi{\'c}, Dmitriy Morozov, and Primoz Skraba.
\newblock Zigzag persistent homology in matrix multiplication time.
\newblock In {\em Proceedings of the Twenty-Seventh Annual Symposium on
  Computational Geometry}, pages 216--225, 2011.

\bibitem{Dionysus2}
Dmitriy Morozov.
\newblock {\tt Dionysus2}.
\newblock URL: \url{https://www.mrzv.org/software/dionysus2/}.

\bibitem{oudot2015zigzag}
Steve~Y. Oudot and Donald~R. Sheehy.
\newblock Zigzag zoology: Rips zigzags for homology inference.
\newblock {\em Foundations of Computational Mathematics}, 15(5):1151--1186,
  2015.

\bibitem{gudhi:urm}
{The GUDHI Project}.
\newblock {\em {GUDHI} User and Reference Manual}.
\newblock {GUDHI Editorial Board}, 2015.
\newblock URL: \url{http://gudhi.gforge.inria.fr/doc/latest/}.

\bibitem{zhang2020gpu}
Simon Zhang, Mengbai Xiao, and Hao Wang.
\newblock {GPU}-accelerated computation of {V}ietoris-{R}ips persistence
  barcodes.
\newblock In {\em 36th International Symposium on Computational Geometry (SoCG
  2020)}. Schloss Dagstuhl-Leibniz-Zentrum f{\"u}r Informatik, 2020.

\bibitem{zomorodian2005computing}
Afra Zomorodian and Gunnar Carlsson.
\newblock Computing persistent homology.
\newblock {\em Discrete \& Computational Geometry}, 33(2):249--274, 2005.

\end{thebibliography}

\appendix

\end{document}